\documentclass[11pt,numbers]{article}
\usepackage[T1]{fontenc}
\usepackage[utf8]{inputenc}
\usepackage[a4paper,lmargin=2cm,rmargin=2cm,tmargin=2cm,bmargin=1.5cm]{geometry}
\usepackage{mathptmx}
\usepackage{cancel}
\usepackage[title]{appendix}
\usepackage[square,numbers]{natbib}

\usepackage{authblk}

\usepackage{longtable}
\usepackage{multirow}
\usepackage{array}
\usepackage{graphicx}
\usepackage{enumerate}
\usepackage{float}
\usepackage{mathrsfs}
\usepackage{tikz-cd}

\usepackage{amsmath}
\usepackage{amssymb}
\usepackage{amsfonts}
\usepackage{amsthm}
\usepackage{mathtools}
\usepackage{color}
\usepackage{thmtools}
\usepackage{braket}

\declaretheoremstyle[
	spaceabove=6pt, spacebelow=6pt,
	headfont=\normalfont\bfseries,
	notefont=\mdseries\bfseries, 
	notebraces={-- }{},
	headpunct=,
	bodyfont=\normalfont,
	postheadspace=\newline
]{mystyle}		

\declaretheorem[
	style=mystyle,
	name=Definition,
	numbered=yes,
	numberwithin=section]{defn}

\declaretheorem[style=mystyle,name=Proposition]{prop}
\declaretheorem[style=mystyle,name=Corollary]{corol}
\declaretheorem[style=mystyle,name=Example]{example}

\usepackage{hyperref}
\hypersetup{
	colorlinks=true,
	linkcolor=blue,
	filecolor=magenta,      
	urlcolor=cyan,
}

\usepackage[textsize=small]{todonotes}
\newcommand{\KP}[1]{{\color{orange}  #1}}

\renewcommand{\int}[1]{[#1]}

\DeclarePairedDelimiter\ceil{\lceil}{\rceil}
\DeclarePairedDelimiter\floor{\lfloor}{\rfloor}

\newcommand\restr[2]{{
  \left.\kern-\nulldelimiterspace 
  #1 
  \vphantom{\big|} 
  \right|_{#2} 
  }}

\begin{document}

\title{Solving decision problems by distributed consensus with one-dimensional, binary, radius-$2$ cellular automata over cyclic configurations}
\date{}

\author[ ]{\footnote{Corresponding author.}{Eurico Ruivo$^{1}$}}
\author[ ]{Pedro Paulo Balbi$^{1}$}
\author[ ]{Kévin Perrot$^{2}$}
\author[ ]{Marco Montalva-Medel$^{3}$}
\author[ ]{Eric Goles$^{3}$}
\affil[1]{Universidade Presbiteriana Mackenzie, Faculdade de Computa\c{c}\~{a}o e Inform\'{a}tica}
\affil[ ]{Rua da Consola\c{c}\~{a}o 896, Consola\c{c}\~{a}o;  01302-907 S\~{a}o Paulo, SP, Brazil}
\affil[2]{Aix Marseille Univ, CNRS, LIS, Marseille, France}
\affil[ ]{Avenue de Luminy 163; F-13288 Marseille Cedex 9, France}
\affil[3]{Universidad Adolfo Ib\'{a}\~{n}ez, Facultad de Ingenier\'{i}a y Ciencias}
\affil[ ]{Av. Diagonal Las Torres 2700, 7941169, Pe\~{n}alol\'{e}n, Santiago, Chile}
\affil[ ]{\texttt{\small eurico.ruivo@mackenzie.br, pedrob@mackenzie.br, kevin.perrot@lis-lab.fr, marco.montalva@uai.cl, eric.goles@uai.cl}}

\maketitle

\begin{abstract}
Probing the ability of automata networks to solve decision problems has received a continuous attention in the literature, and specially with the automata reaching the answer by distributed consensus, i.e., their all taking on a same state, out of two. In the case of binary automata networks, regardless of the kind of update employed, the networks should display only two possible attractors, the fixed points $0^L$ and $1^L$, for all cyclic configurations of size $L$. A previous investigation into the space of one-dimensional, binary, radius-2 cellular automata identified a restricted subset of rules as potential solvers of decision problems, but the reported results were incomplete and lacked sufficient detail for replication. To address this gap, we conducted a comprehensive reevaluation of the entire radius-2 rule space, by filtering it with all configuration sizes from 5 to 20, according to their basins of attraction being formed by only the two expected fixed points. A set of over fifty-four thousand potential decision problem solvers were then obtained. Among these, more than forty-five thousand were associated with 3 well-defined decision problems, and precise formal explanations were provided for over forty thousand of them. The remaining candidate rules suggest additional problem classes yet to be fully characterised. Overall, this work substantially extends the understanding of radius-2 cellular automata, offering a more complete picture of their capacity to solve decision problems by consensus.

{\it Keywords}: Distributed consensus, decision problem, cellular automata, cyclic configuration.
\end{abstract}

\section{Introduction}
\label{sec:introduction}

Probing the ability of automata networks, specially cellular automata, to solve decision problems has received a continuous attention in the literature. Given the intrinsic local processing nature of the automata in the network, in order to facilitate the process of reaching the answer of the problem the decision process has taken the form of all automata reaching a distributed consensus, that is, their all taking on a same state, out of two, that would mean one solution or the other.

In the case of binary automata networks, the most usual situation, regardless of the kind of update employed, the networks should display  only two possible attractors, the fixed points $0^L$ and $1^L$, for all cyclic configurations of size $L$.

The efforts in the literature related to automata networks for solving decision problems along this line has focused on two possibilities, according to the kind of decision the network would have to make with respect to the initial configuration: the characterisation of the problem in terms of a formal language the network would be recognising or not \cite{mining2018,portfolio2018}, or in classical benchmark problems, like the determination in the initial configuration of the parity of 1s \cite{Parity150&Shifts2024,correctedBFO_arXiv2025} or the prevailing state in number of occurrences \cite{SurveyDCT_JCA2014,negativeDCT2024}.

With respect to decision problems characterised as formal language recognition, \cite{mining2018} provides complete results in terms of regular languages recognised by one-dimensional, binary cellular automata, with radii 0.5, 1 and 1.5 over cyclic configurations and synchronous update. That reference also addressed the radius 2 case. In fact, it is reported therein that out of the entire radius 2 space (composed of $2^{2^5}$ = 4,294,967,296 rules) computational and analytical pruning processes led to 40,941 candidate rules as possible decision problem solvers, out of which 1,121 have been completely characterised in terms of a same language recognition problem they were related with (namely, the presence of at least one 0 in the initial configuration). However, not enough detail was provided in that reference that would allow for a precise replication of the latter two rule quantities, and a revision we performed on the data employed therein did not lead to any further clarification. 

In face of this, the results in \cite{mining2018} required a reevaluation and this is exactly what the present paper is all about. In order to go about it we started from scratch, with all $2^{2^5}$ rules of the space, only trimmed by taking the representative of the classes of dynamical equivalence -- due to negation, reflection, or the composition of both -- chosen as the one with the smallest Wolfram number in the class. Then we filtered the previous set towards identifying likelier consensus type rules, by computational checking the basins of attraction of all the rules for all configuration sizes $L$ from $L=5$ up to $L=20$ (stricter than in \cite{mining2018}, where it went only up to $L=16$); this left 54,928 candidate rules. Out of the latter, 45,699 of them were identified as candidates for 3 well-defined decision problems, leaving the remaining 9,229 as candidates for other decision problems. Finally, out of the 45,699 candidates in the 3 classes, we provided precise formal explanation about the dynamics of 40,254 of them. So, as will become clear herein, our reevaluation of radius 2 space significantly extended what was known about the space of possible decision problems in the space, by settling the matter on many more rules and at the same time identifying a number of other candidate decision problems.

The rest of the paper is structured as follows. The next section gives all formal definitions required for what will come later. Section~\ref{sec:consensusrules} provides an overview of the results obtained, by identifying 45,699 candidate rules in 3 decision problems. Section~\ref{sec:proofs} is the core of the paper, where in the 3 classes identified, 40,254 of them were proven to truly be of the distributed consensus type. The last section wraps up the results and discusses the perspectives for handling the remaining 14,674 candidate rules and their potential as solvers of more sophisticated problems. An appendix concludes the text with a complete listing of all distributed consensus decision problems derived from the work, both those characterised here and the candidates ones.

\section{Basic definitions}
\label{sec:definitions}

\subsection{Cellular automata}

\begin{defn}[Cellular automaton] A \emph{cellular automaton} (CA) is a 4-tuple $(S,N,f,d)$ in which \cite{Kari2005}:

\begin{itemize}
\item $S=\{0,1,\cdots,k-1\}$, $k\in\mathbb{Z},k\geq 2$ is the $k$-ary \emph{set of states};
\item $N=\left(\vec{v}_1,\cdots,\vec{v}_{|N|}\right)$, $\vec{v}_i\in \mathbb{Z}^d,\forall 1\leq i\leq |N|$, is the \emph{neighbourhood vector};
\item $f:S^{|N|}\rightarrow S$ is the \emph{local transition function} or \emph{local rule};
\item $d\in\mathbb{Z},d\geq 1,$ is the \emph{dimension}.
\end{itemize}

When $S=\{0,1\}$, $d=1$ and $N=\{-\ceil{r},\cdots,0,\cdots\floor{r}\}$, with $r\in\{\frac{1}{2},1,\frac{3}{2},2,\frac{5}{2},\cdots\}$, the CA is said to be a \emph{radius-}$r$, binary, one-dimensional CA. In the remainder of the paper we will only refer to one-dimensional binary CAs and, for the sake of simplicity, such CAs will be referred to simply as \emph{radius}-$r$ CAs. In particular, a radius-$1$ binary one-dimensional CA is an \emph{elementary cellular automaton} (ECA).

Given a radius-$r$ local rule $f$, its \emph{(Wolfram) rule number} \cite{Wolfram2002} is given by

\begin{equation*}
W(f)=\sum_{x=(x_1,\cdots,x_{2r+1})\in\{0,1\}^{2r+1}}f(x)\cdot 2^{\left( x_1\cdot 2^{2r}+x_2\cdot 2^{2r-1}+\cdots+x_{2r+1}\cdot 2^0\right)}
\end{equation*}

For instance, the local ECA rule given by $f(x_1,x_2,x_3)=(x_1+x_2+x_3)\mbox{ mod }2$ has Wolfram number 150.

\end{defn}

\begin{defn}[Configurations]

A \emph{configuration} is a function $c:\mathbb{Z}\rightarrow \{0,1\}$. A \emph{cell} of a configuration is an integer $i$ and $c(i)$ is the \emph{state of cell} $i$. From now on, $c(i)$ will be denoted by $c_i$. The set of all configurations is denoted by $\mathscr{C}$.
\end{defn}

Given a radius-$r$ CA with local rule $f$, such a local rule induces a \emph{global transition function} (or \emph{global rule}) $F:\mathscr{C}\rightarrow \mathscr{C}$ in the set of all configurations, as follows:

\begin{equation*}
\left(F(c)\right)_i=f(c_{i-\ceil{r}},\cdots,c_i,\cdots,c_{i+\floor{r}})
\end{equation*}

In order to ease the notation, $[i,j]$, with $i<j\in\mathbb{Z}$ will denote the set of integers from $i$ to $j$. Then we can write $\left(F(c)\right)_i=f(c_{[i-\ceil{r},i+\floor{r}]})$.

The \emph{orbit of} $c$ under $f$ is the sequence $\left(c,F(c),F^2(c),F^3(c),\cdots\right)$, in which $F^n,n\geq 1$, denotes the iteration of $F$ $n$ times over $c$.

A \emph{configuration of length } $L$ (or a \emph{length} $L$ \emph{configuration}), $L\in\mathbb{Z},L\geq 1$, is a configuration such that $c_i=c_{i+L}$ for all $i\in\mathbb{Z}$. Such configurations can be regarded as vectors $c=(c_0,\cdots,c_{L-1})\in\{0,1\}^L$ and the set of all such vectors is denoted by $\mathscr{C}_L$.

The local function $f$ induces a global function $F_L:\mathscr{C}_L\rightarrow\mathscr{C}_L$ in the set of all configurations of length $L$ as follows:

\begin{equation*}
\left(F_L(c)\right)_i=f(c_{(i-\ceil{r})\mbox{ mod }L},\cdots,c_i,\cdots,c_{(i+\floor{r})\mbox{ mod }L})
\end{equation*}

From now on, we will only address configurations of length $L$, therefore, for the sake of simplicity, $F_L$ will be denoted simply by $F$ and the indices of the cells are always assumed to be taken modulo $L$.

A \emph{block of length} $l\geq 1$ is a vector $b=(b_1,\cdots,b_{l})\in\{0,1\}^l$. Such block can also be denoted by $b=b_1\cdots b_{l}$.

Given two blocks $b=b_1\cdots b_l$ and $\hat{b}=\hat{b}_1\cdots\hat{b}_{\hat{l}}$, with $l\leq \hat{l}$, we say that $b$ is a \emph{sub-block of} $\hat{b}$ if there is $1\leq j\leq \hat{l}$ such that $b=\hat{b}_{[j,j+l-1]}$.

Given a configuration $c$ of length $L\geq 1$, we say that a block $b=b_1\cdots b_l$ of length $l\leq L$ is a \emph{sub-block of }$c$ if there is $0\leq j\leq L-1$ such that $b=c_{[j,j+l-1]}$.

Given a block $b=b_1\cdots b_l$, $b^n$ will denote the concatenation of block $b$ $n$ times, that is

\begin{equation*}
b^n=\underbrace{b_1\cdots b_l}_{1}\underbrace{b_1\cdots b_l}_{2}\cdots\underbrace{b_1\cdots b_l}_{n}
\end{equation*}

\begin{defn}[Extension of $f$ to blocks]

We can extend the application of a local function of radius-$r$ to a block of length $l\geq 2r+1$ as follows. Let $b=b_0\cdots b_{l-1}$ a block, and consider the function $\hat{f_l}:\{0,1\}^l\rightarrow\{0,1\}^{l-2r}$ defined by:

\begin{equation*}
\hat{f_l}(b_0\cdots b_{l-1})=f(b_{[0,2r]})f(b_{[1,2r+1]})\cdots f(b_{[l-2r-1,l-1]})
\end{equation*}
\end{defn}

The function $\hat{f_l}$ is the \emph{extension of} $f$ \emph{to blocks of length} $l$. When $l$ is clear from the context, we will drop it from the notation and write simply $\hat{f}$.

\begin{example}[Image of a block under an extended local function]
Let $f:\{0,1\}^3\rightarrow \{0,1\}$ be a radius-$1$ CA local rule given by $f(x,y,z)=(x+y+z)\text{ mod }2$. The image of block $110100$ under $\hat{f}_6\dot{\equiv} \hat{f}$ is given by:

\begin{equation*}
\hat{f}(110100)=f(110)f(101)f(010)f(100)=0011
\end{equation*}
\end{example}











\subsection{Attractors}

\begin{defn}[Attractors]
\label{defn:attractors}
Let $f$ be a radius-$r$ CA local rule and let $X=\left\{x^{(1)},x^{(2)},\cdots,x^{(n)}\right\}$ be a set of $n$ configurations of length $L$. The set $X$ is said to be an \emph{attractor of length} $L$ \emph{of} $F$ when 

\begin{itemize}

\item $F(X)=X$, with $F(X)=\left\{F\left(x^{(1)}\right),\cdots,F\left(x^{(n)}\right)\right\}$;

\item If $Y\subseteq X$ and $F(Y)=Y$, then $Y=X$.
\end{itemize}

If $n=1$, such attractor is named a \emph{fixed-point} and it will be denoted simply as $x^{(1)}$ instead of $\{x^{(1)}\}$ and we will say that $x^{(1)}$ is an attractor.

If $n>1$, the attractor is named a \emph{limit-cycle} and, for any $x^{(i)}\in X$, we have $F^n\left(x^{(i)}\right)=x^{(i)}$. In that case, $n$ is a \emph{period} of $X$, as well as any positive multiple of $n$ is also a period of $X$.

Let $X$ be an attractor of length $L$ and let $B_X=\left\{c\in\mathscr{C}_L:F^t(c)\in X \mbox{ for some }t\geq 0\right\}$. That is, $B_X$ is the set of configurations of length $L$ such that, after being iterated under $F$ a certain number of times, the resulting image end up belonging to $X$. The set $B_X$ is the \emph{basin of attraction of} $X$.
\end{defn}

\begin{defn}[Consensus rules]
\label{defn:classificationrule}

We say that a (binary, one-dimensional) radius-$r$ CA rule is a \emph{consensus rule} (or a \emph{decision rule}) if, for any $L\geq 2r+1$, its only attractors are $0^L$ and $1^L$. In this case, for each $L$, the basin of attraction of $0^L$ is named \emph{Class 0} and the basin of attraction of $1^L$ is named \emph{Class 1}.

In other words, consensus rules partition each $\mathscr{C}_L$ into two classes: the configurations converging to $0^L$ and the configurations converging to $1^L$. For simplicity, the attraction basin of $0^L$ (respectively, the basin of attraction of $1^L$) for a rule $f$ will be denoted by $\mathscr{B}_{(f,0)}^L$ (respectively, by $\mathscr{B}_{(f,1)}^L$).

Also, $\mathscr{B}_{(f,x)}=\bigcup_{L\geq 2r+1}\left(\mathscr{B}_{(f,x)}^L\right),x\in\{0,1\}$, will denote the union of all basins of attraction of $x^L$ for different values of $L\geq 2r+1$.

If $f$ is clear from the context, we will write simply $\mathscr{B}_x$ and $\mathscr{B}^L_x$, respectively, instead of $\mathscr{B}_{(f,x)}$ and $\mathscr{B}^L_{(f,x)}$.

\end{defn}

For example, a solution to the parity consensus would have

\begin{itemize}
\item $\mathscr{B}_0^L=\{ c \in \mathscr{C}_L : c \text{ has an even number of cells in state } 1 \}$ and
\item $\mathscr{B}_1^L=\{ c \in \mathscr{C}_L : c \text{ has an odd number of cells in state } 1 \}$,
\end{itemize}
\noindent and a solution to the majority consensus would have
$\mathscr{B}_x^L=\{ c \in \mathscr{C}_L : c \text{ has a majority of cells in state } x \}$ for $x\in\{0,1\}$.

In the next section we discuss how the rules studied in this paper were obtained and how they are classified in terms of the decision problem they solve for the synchronous update.

\section{Radius-2 consensus rules}
\label{sec:consensusrules}

As explained in Section~\ref{sec:introduction}, here we revisited the data and analysis available in the literature for binary, radius-$2$, one-dimensional cellular automata, with synchronous updates, over cyclic configurations.

As such, recall that we started from scratch, with all $2^{2^5}$ rules of the space, partitioned the set into its classes of dynamical equivalence (due to the symmetries derived from negation, reflection, or the composition of both) and considered the representative rule of each class as the one with the smallest Wolfram number. Then we filtered the resulting rule set towards identifying likelier consensus type rules, by computational checking the basins of attraction of all the rules for all configuration sizes $L$ from $L=5$ up to $L=20$, so that they would only feature as attractors the fixed points $0^L$ and $1^L$. This led to 54,928 candidate rules. 

Out of the latter, 45,699 of them were identified as candidates for 3 well-defined decision problems, as listed right below, organised into 3 main rule classes (A, B and C), comprising 83.2\% of the original set of candidate rules. 

It should be noticed that the remaining set of 9,229 uncharacterised rules comprise three possibilities: rules that may still come to be associated to the decision problems of the 3 classes, but were left out because the characterisation we provided (to be presented in the next section) may have been incomplete; rules associated to additional decision problems, yet to be characterised (see the extensive discussion and examples in Subsection~\ref{subsec:discussion} and in the Appendix); or rules that may not be candidate decision problems at all, but were not filtered out since they were subjected only to configuration sizes up to $L=20$. 

Finally, out of the 45,699 candidates in the 3 classes, we provided precise formal explanation about the dynamics of 40,254 of them, or 88.1\%, meaning that the remaining set of 5,445 rules should be added to the previous 9,229 candidate rules, thus totalling 14,674 rules, or 26.7\% of all rules obtained out of the filtering process.

Summing up, out of the 54,928 radius-$2$ consensus rules obtained for $5 \leq L \leq 20$, at least 40,254 rules (representing, approximately, $73.3 \%$ of the total amount of potential consensus rules) really do solve decision problems for configurations of any length $L \geq 5$. So, as detailed below, our results on the radius 2 space significantly extended what was known about the space of possible decision problems in the space, by settling the matter on many more rules and at the same time identifying a number of other candidate decision problems.

\begin{itemize}
\item \textbf{Class A} (30,230 rules): $\mathscr{B}^L_{(f,1)}=\{1^L\}$
\item \textbf{Class B} (14,680 rules):
    $
        \mathscr{B}^L_{(f,1)}  =  
        \begin{cases}
        \{1^L\},\mbox{ if }L\mbox{ is odd;}\\
        \{1^L,(01)^{L/2},(10)^{L/2}\},\mbox{ if }L\mbox{ is even}
        \end{cases}
    $
\item \textbf{Class C} (789 rules): 
    $
           \mathscr{B}^L_{(f,1)}  =  
           \begin{cases}
        \{1^L\},\mbox{ if }L\mbox{ is odd;}\\
        \{1^L,(01)^{L/2},(10)^{L/2}\},\mbox{ if }L\equiv 2 \mbox{ mod }4;\\
         \{1^L,(01)^{L/2},(10)^{L/2},(1100)^{L/4},(1001)^{L/4},(0011)^{L/4},(0110)^{L/4}\},\\\mbox{ if }L\equiv 0 \mbox{ mod }4;
        \end{cases}
    $

\end{itemize}

Notice that, although the rules in each class have the same attractors (at least up to $L=20$), the actual dynamics of the rules are not the same, i.e., their basins of attraction are different.

In the next section we show that most rules in Classes A, B and C are indeed consensus rules which have exactly the basins of attraction described above for all $L\geq 5$.

\section{Proofs}
\label{sec:proofs}

Each subsection below addresses one of the rule classes and is divided in two parts: first, we provide the proofs regarding the basins of attraction for rules in each class; then, we describe how the properties of the rules in each class were checked.

\subsection{Class A}
\label{subsec:0forcing}

\subsubsection{Proofs for Class A}

Let $f$ be a local rule in Class A. Then either $f$ is a 0-forcing rule or there is an integer $m$ such that $f^m$ is a 0-forcing rule. The results below show that in either case, there are only two attractors, namely $0^L$ and  $1^L$, and that $\mathscr{B}_1=\{1^L\}$ and $\mathscr{B}_0=\mathscr{C}_L-\{1^L\}$.

In other words, the results below show that rules in Class A simply detect whether or not there is at least one cell in state 0 in a configuration of length $L$.

\begin{prop}
\label{prop:x0x1}
Let $f$ be a binary radius-$r$ CA local rule such that $f\left (x_{\left[-\ceil{r},\floor{r}\right]}\right)=x_0\cdot x_i\cdot \phi\left( x_{\left[-\ceil{r},\floor{r}\right]}\right)$, with $i\in\{-1,1\}$, for some function $\phi: \{0,1\}^{2r+1}\rightarrow \{0,1\}$ such that $\phi(1^{2r+1})=1$. Then the only attractors of length $L\geq 2r+1$ of the induced global rule $F$ are $0^L$ and $1^L$ and they are fixed-points. Moreover, if a configuration $c\in \mathscr{C}_L$ has any of its cells in state 0, then $c\in\mathscr{B}^L_0$.
\end{prop}

\begin{proof}
First, notice that $f(1^{2r+1})=1\cdot 1\cdot \phi(1^{2r+1})=1$ and $f(0^{2r+1})=0\cdot 0\cdot \phi(0^{2r+1})=0$. Therefore $F(0^L)=0^L$ and $F(1^L)=1^L$, for all $L\geq 2r+1$.

Let $c$ be a configuration of length $L\geq 2r+1$ such that there is $0\leq j\leq L-1$ with $c_j=0$.

For $i=-1$, we will show that $\left(F^k(c)\right)_{[j,j+k]}=0^{k+1}$ for all $k\geq 0$. For $k=0$, it is trivial, since $c_j=0$ by hypothesis. Suppose $\left(F^l(c)\right)_{[j,j+l]}=0^{l+1}$ for some $l\geq 0$. Then, we have:

\begin{eqnarray*}
\left(F^{l+1}(c)\right)_{j+m}& = &F\left(F^l(c)\right)_{j+m}=f(F^l(c)_{j+m-\ceil{r}},\cdots,F^l(c)_{j+m-1},0,F^l(c)_{j+m+1},\cdots,F^l(c)_{j+m+\floor{r}})=\\
& = & 0\cdot F^l(c)_{j+m-1}\cdot \phi(F^l(c)_{[j+m-\ceil{r},j+m+\floor{r}]})=0 \mbox{ for all } 0\leq m\leq l, \mbox{ and }\\
\left(F^{l+1}(c)\right)_{j+l+1}& = &F\left(F^l(c)\right)_{j+l+1}=f(F^l(c)_{j+l+1-\ceil{r}},\cdots,F^l(c)_{j+l-1},0,F^l(c)_{j+l+1},\cdots,F^l(c)_{j+l+1+\floor{r}})=\\
& = & F^l(c)_{j+l+1}\cdot 0 \cdot \phi(F^l(c)_{[j+l+1-\ceil{r},j+l+1+\floor{r}]})=0
\end{eqnarray*}

By induction, it follows that $\left(F^k(c)\right)_{[j,j+k]}=0$ for all $k\geq 0$. In particular for $k=L-1$, this implies $F^{L-1}(c)=0^L$, which is a fixed-point. Therefore, $c\in\mathscr{B}^L_0$.

In other words, $f$ decides whether or not a configuration has at least one cell in state 0: if it does, it converges to $0^L$ under $f$; if it does not, then it remains fixed in $1^L$.

The proof for $i=1$ is analogous, sufficing to assume $\left(F^l(c)\right)_{[j-l,j]}=0^{l+1}$ for some $l\geq 0$ in the induction step and show that $\left(F^k(c)\right)_{[j-k,j]}=0^{k+1}$ for all $k\geq 0$.

\end{proof}

\begin{corol}
\label{corol:shifted}

Let $f$ be a binary radius-$r$ one-dimensional CA rule such that $f(x_{[-\ceil{r},\floor{r}]})=x_{i}\cdot x_{i+1}\cdot \phi(x_{[-\ceil{r},\floor{r}]})$ for some function $\phi:\{0,1\}^{2r+1}\longrightarrow \{0,1\}$ such that $\phi(1^{2r+1})=1$ and some $-\ceil{r}\leq i\leq \floor{r}$. Then the only attractors of length $L$ of $F$ are $0^L$ and $1^L$. Moreover, if $c\in\mathscr{C}_L$ has any of its cells in state 0, then $c\in\mathscr{B}^L_0$.

\end{corol}
Let $L\geq 2r+1$.

First, notice that $F(0^L)=f(0^{2r+1})^L=\left[0\cdot 0\cdot \phi(0^{2r+1})\right]^L=0^L$ and $F(1^L)=f(1^{2r+1})^L=\left[1\cdot 1\cdot \phi(1^{2r+1})\right]^L=1^L$. That is, $0^L$ and $1^L$ are fixed-points of length $L$ of $F$. In particular, $\{1^L\}\subseteq \mathscr{B}^L_1$.

The case $i=0$ is already proven in Proposition \ref{prop:x0x1}.

Let $g$ be a radius-$(r+i)$ CA local rule given by
$g(x_{[-\ceil{r}-|i|,\floor{r}+|i|]})=f(x_{[-\ceil{r}-|i|,\floor{r}-|i|]})$.

Define $\tilde{\phi}:\{0,1\}^{2(r+|i|)+1}\longrightarrow \{0,1\}$ by $\tilde{\phi}(x_{[-\ceil{r}-|i|,\floor{r}+|i|]})=\phi(x_{[-\ceil{r}-i,\floor{r}-i]})$. That is, $g$ is the $(-i)$-shift of $f$ and $G=\sigma_{-i}\circ F$.

Notice that:
\begin{equation*}
g(x_{[-\ceil{r}-|i|,\floor{r}+|i|]})=f(x_{[-\ceil{r}-i,\floor{r}-i]})=
x_0\cdot x_1\cdot \phi(x_{[-\ceil{r}-i,\floor{r}-i]})= x_0\cdot x_1\cdot\tilde{\phi}(x_{-\ceil{r}-|i|,\floor{r}+|i|}).
\end{equation*}

Also, since $\tilde{\phi}(1^{2(r+|i|)+1})=\phi(1^{2r+1})=1$, $g$ satisfies the conditions in Proposition \ref{prop:x0x1}, therefore $0^L$ and $1^L$ are the only attractors of length $L$ of $G$.

Suppose $F$ has a fixed-point $x$ of length $L$ such that $x\notin\{0^L,1^L\}$. Then we would have:
\begin{equation*}
G^L(x)=\left(\sigma_{(-i)}\circ F\right)^L(x) = \left(\sigma_{(-i)}\right)^L\left( F^L(x)\right)=\left(\sigma_{-L\cdot i}\right)(x)=\sigma_0(x)=x.
\end{equation*}

That is, $x$ would belong to a limit-cycle of period $L$ of $G$, contradicting the fact that $0^L$ and $1^L$ are the only attractors of $G$.

On the other hand, if $F$ has a configuration $x$ of length $L$, $x\notin \{0^L,1^L\}$, belonging to a limit-cycle of period $p>1$, then we would have $F^p(x)=x$ and:
\begin{equation*}
G^{L\cdot p}(x)=\left(\sigma_{(-i)}\circ F\right)^{L\cdot p}(x) = \left(\sigma_{(-i)}\right)^{L\cdot p}\left( F^{L\cdot p}(x)\right)=\left(\sigma_{-L\cdot p\cdot i}\right)(x)=\sigma_0(x)=x.
\end{equation*}

In that case, $x$ would belong to a limit-cycle of period $L\cdot p$ of $G$, contradicting the fact that $0^L$ and $1^L$ are the only attractors of 
$G$.

Therefore, 
the only attractors of length $L$ of $F$ are $0^L$ and $1^L$.

Moreover, assume $c\in\mathscr{C}_L$ is such that $c_j=0$ for some $0\leq j\leq L-1$ and suppose $F^k(c)=1^L$ for some integer $k>0$. Then we would have:
\begin{equation*}
G^k(c)=(\sigma_{(-i)})^k\left((F^k)(c)\right)=\sigma_{(-k\cdot i)}(1^L)=1^L.
\end{equation*}

In other words, $c$ would belong to the basin of attraction of $1^L$ for $G$, which is a contradiction, since $g$ satisfies the conditions in Proposition \ref{prop:x0x1}. Therefore, $c$ is in the basin of attraction of $0^L$ for $F$.



\begin{prop}
\label{prop:x0x2spread}
Let $f$ be a binary radius-$r$ CA local rule such that $f\left (x_{\left[-\ceil{r},\floor{r}\right]}\right)=x_0\cdot x_i\cdot \phi\left( x_{\left[-\ceil{r},\floor{r}\right]}\right)$, with $i\in\{-2,2\}$, for some function $\phi: \{0,1\}^{2r+1}\rightarrow \{0,1\}$. Let $c$ be a configuration of length $L\geq 2r+1$ such that $c_j=0$ for some $0\leq j\leq L-1$.

Then $\left(F^k(c)\right)_{j+2l}=0$ for all $0\leq l\leq k$ and all $k\geq 0$.
\end{prop}
\begin{proof}
First notice that, for any configuration $u$, if $u_p=0$ for some $0\leq p\leq L-1$, then $F(u)_p=0$, since $F(u)_p=f(u_{p-\ceil{r}},\cdots,u_{p-1},0,u_{p+1},\cdots,u_{p+\floor{r}})=0\cdot u_{p+\KP{i}}\cdot \phi\left(u_{[p-\ceil{r},p+\floor{r}]}\right)=0$. In other words, cells in state $0$ remain in state $0$.

Assume $i=-2$. Since the $0$s are fixed under $F$, it suffices to show that $\left(F^k(c)\right)_{j+2k}=0$ for all $k\geq 0$. For $k=0$ it is trivial, since $c_j=0$ by hypothesis. Assume $\left(F^\lambda(c)\right)_{j+2\lambda}=0$ for some $\lambda\geq 0$ and let $y=F^\lambda (c)$. Then:

\begin{eqnarray*}
\left(F^{\lambda+1}(c)\right)_{j+2(\lambda+1)} & = & \left(F(y)\right)_{j+2(\lambda +1)}=f\left(y_{j+2(\lambda+1)-\ceil{r}},\cdots,y_{j+2\lambda-1},0,y_{j+2\lambda +1},\cdots,f_{j+2(\lambda+1)+\floor{r}}\right) = \\
& = & y_{j+2(\lambda+1)} \cdot 0 \cdot \phi(y_{[j+2(\lambda+1)-\ceil{r},j+2(\lambda+1)+\floor{r}]})=0
\end{eqnarray*}

By induction over $\lambda$ and the fact that all 0s remain fixed throughout the evolution of $f$, it follows that $\left(F^k(c)\right)_{j+2l}=0$ for all $0\leq l\leq k$ and all $k\geq 0$.

The proof for $i=2$ is analogous.
\end{proof}

\begin{corol}
\label{corol:x0x2spread}
Let $f$ be a binary radius-$r$ CA local rule such that $f\left (x_{\left[-\ceil{r},\floor{r}\right]}\right)=x_0\cdot x_i\cdot \phi\left( x_{\left[-\ceil{r},\floor{r}\right]}\right)$, with $i\in\{-2,2\}$, for some function $\phi: \{0,1\}^{2r+1}\rightarrow \{0,1\}$. Let $c$ be a configuration of length $L\geq 2r+1$ such that $c_j=0$ for some $0\leq j\leq L-1$. Then:
\begin{enumerate}
\item If $L$ is even, then $\left(F^{\frac{L-2}{2}}(c)\right)_{l}=0$, for all $l\equiv j\mbox{ mod }2$;
\item If $L$ is odd, then $F^{\frac{L-1}{2}}(c)$ has, at least, two adjacent cells in state 0.
\end{enumerate}
\end{corol}

\begin{proof}
Suppose $i=-2$.

1.
Assume $L=2k$ for some integer $k\geq r+1$. By Proposition \ref{prop:x0x2spread}, $\left( F^{\frac{L-2}{2}}(c)\right)_{j+2l}=0$ for all $0\leq l\leq \frac{L-2}{2}=k-1$. In other words, cells $\{j,j+2,j+4,\cdots,j+2k-2\}$ are in state $0$ after $\frac{L-2}{2}$ iterations of $F$ over $c$. Notice that such cells comprise every cell in $c$ which position have the same parity as $j$.

2.
Assume $L=2k+1$ for some integer $k\geq r$. By Proposition \ref{prop:x0x2spread}, $\left( F^{\frac{L-1}{2}}(c)\right)_{j+2l}=0$ for all $0\leq l\leq \frac{L-1}{2}=k$. In particular, $\left(F^{\frac{L-1}{2}}(c)\right)_{j}=0$ and $\left(F^{\frac{L-1}{2}}(c)\right)_{j+2k}=0$.

Notice that cells $j$ and $j+2k=j+L-1$ are adjacent, since by taking the positions modulo $L$ we have $(j+L-1)\mbox{ mod }L = j-1$, which is the immediate cell to the left of cell $j$.

The proof is analogous for $i=2$.
\end{proof}

\begin{prop}
\label{prop:x0x2adjacent0s}
Let $f$ be a binary radius-$r$ CA local rule such that $f\left (x_{\left[-\ceil{r},\floor{r}\right]}\right)=x_0\cdot x_i\cdot \phi\left( x_{\left[-\ceil{r},\floor{r}\right]}\right)$, with $i\in\{-2,2\}$, for some function $\phi: \{0,1\}^{2r+1}\rightarrow \{0,1\}$. Let $c$ be a configuration of length $L\geq 2r+1$ such that $c_j=c_{j+1}=0$ for some $0\leq j\leq L-1$. Then $F^k(c)=0^L$ for all $k\geq \floor{\frac{L}{2}}$.
\end{prop}

\begin{proof}
First notice that $0^L$ is a fixed-point of $F$, since $f(0,\cdots,0)=0\cdot 0\cdot \phi(0,\cdots,0)=0$.

If $L$ is even, by applying Corollary \ref{corol:x0x2spread} for cells $j$ and $j+1$, $\left(F^{\frac{L-2}{2}}(c)\right)_l$ = 0 for all $l\equiv j\mbox{ mod }2$ and all $l\equiv (j+1)\mbox{ mod 2}$. Since $j$ and $j+1$ have distinct parities, $\left(F^{\frac{L-2}{2}}(c)\right)_l=0$ for all $0\leq l\leq L-1$, that is, $F^{\frac{L-2}{2}}(c)=0^L$.

Also, since $0^L$ is a fixed-point, $F^{k}(c)=0^L$ for all $k\geq \frac{L-2}{2}$ and, with more reason, for all $k\geq \floor{\frac{L}{2}}=\frac{L}{2}$.

If, on the other hand, $L=2p+1$ is odd, for some $p\geq r$, by Proposition \ref{prop:x0x2spread}, $F^{\floor{\frac{L}{2}}}(c)=F^{p}(c)$ is such that:
\begin{itemize}
    \item Cells $\{j, j+2, \cdots,j+2p\}$ are in state 0 (by applying Proposition \ref{prop:x0x2spread} for cell $j$ in $c$);
    \item Cells $\{j+1, j+3,\cdots,j+2p+1\}$ are in state 0 (by applying Proposition \ref{prop:x0x2spread} for cell $j+1$ in $c$).
\end{itemize}

Notice that $\{j, j+2, \cdots,j+2p\}\cup \{j+1, j+3,\cdots,j+2p+1\} = \{j,j+1,j+2,\cdots,j+2p,j+2p+1\}$, which corresponds to every cell in $c$. Therefore $F^{\floor{\frac{L}{2}}}(c)=0^L$ and, since $0^L$ is a fixed-point, $F^k(c)=0^L$ for all $k\geq \floor{\frac{L}{2}}$.

\end{proof}

\begin{corol}
\label{corol:x0x2}
Let $f$ be a binary radius-$r$ CA local rule such that $f\left (x_{\left[-\ceil{r},\floor{r}\right]}\right)=x_0\cdot x_i\cdot \phi\left( x_{\left[-\ceil{r},\floor{r}\right]}\right)$, with $i\in\{-2,2\}$, for some function $\phi: \{0,1\}^{2r+1}\rightarrow \{0,1\}$. 

Let 
$
u_r=
\begin{cases}
(10)^{\frac{r}{2}}1(01)^{\frac{r}{2}},\mbox{ if }r\mbox{ is even};\\
(01)^{\frac{r-1}{2}}010(10)^{\frac{r-1}{2}},\mbox{ if }r\mbox{ is odd};\\
(10)^{\frac{\ceil{r}}{2}}10(10)^{\frac{\ceil{r}}{2}-1},\mbox{ if }r\notin \mathbb{Z}\mbox{ and }\ceil{r}\mbox{is even};\\
(01)^{\frac{\ceil{r}-1}{2}}01(01)^{\frac{\ceil{r}-1}{2}},\mbox{ if }r\notin \mathbb{Z}\mbox{ and }\ceil{r}\mbox{is odd}
\end{cases}
.
$

If $\phi(u_r)=0$ and $\phi(1^{2r+1})=1$, then the only attractors of length $L\geq 2r+1$ of $F$ are $0^L$ and $1^L$. Moreover, if $c\in\mathscr{C}_L$ has any of its cells in state 0, then $c\in\mathscr{B}^L_0$. 
\end{corol}

\begin{proof}
Let $L=2r+1$ and let $c$ be a configuration of length $L$.

If there are no cells in $c$ in state 0, then $F(c)=F(1^L)=\left(f(1^{2r+1})\right)^L=\left(1\cdot 1\cdot\phi(1^{2r+1})\right)=1^L$, showing that $1^L$ is a fixed-point of $F$.

Also, if there are two adjacent cells in $c$ in state 0, by Proposition \ref{prop:x0x2adjacent0s}, then $F^k(c)=0^L$ for all $k\geq \floor{\frac{L}{2}}$, $0^L$ is a fixed-point and $c\in\mathscr{B}^L_0$.

Suppose there is $0\leq j\leq L-1$ such that $c_j=0$ but there are no adjacent cells in state 0 in $c$

If $L$ is odd, by Corollary \ref{corol:x0x2spread}, $F^{\frac{L-1}{2}}(c)$ has two adjacent cells in state 0. Therefore, by Proposition \ref{prop:x0x2adjacent0s}, $F^k\left(F^{\frac{L-1}{2}}(c)\right)=0^L$ for all $k\geq \floor{\frac{L}{2}}$. Hence, $0^L$ is a fixed-point of $F$ and $c\in\mathscr{B}^L_0$.

On the other hand, if $L$ is even, let $y=F^{\frac{L-2}{2}}(c)$. If $y$ has two adjacent cells in state 0, by Proposition \ref{prop:x0x2adjacent0s}, $F^k(y)=F^{k+\frac{L}{2}}(c)=0^L$ for all $k\geq \floor{\frac{L}{2}}$, $0^L$ is a fixed-point of $F$ and $c\in\mathscr{B}^L_0$.

Also, if $y$ does not have two adjacent cells in state 0, then $y$ is composed by alternating 0s and 1s, that is, $y=(10)^{\frac{L}{2}}$. In this case, since $\phi(u_r)=0$ by hypothesis, we have:

\begin{equation*}
\left(F(y)\right)_j=y_{j}\cdot y_{j+i}\cdot\phi(u_r)=y_{j}\cdot y_{j+i}\cdot 0=0,\mbox{ for all }0\leq j\leq L-1.
\end{equation*}

In other words, $F^{\frac{L-2}{2}+1}(c)=0^L$ and, since $F(0^L)=0^L$, $0^L$ is also a fixed-point of $F$ in this case  and $c\in\mathscr{B}^L_0$.

Therefore, if configuration $c$ has at least one cell in state 0, it converges to $0^L$. If it has no cells in state 0, then it is already the fixed-point $1^L$.
\end{proof}

In summary, the decision problem solved by rules described in the statement of Corollary~\ref{corol:x0x2} is simply whether or not a configuration has at least one cell in state 0.

\begin{prop}
\label{prop:onlyattractors}
Let $L\geq 2r+1$ and let $f$ be a binary radius-$r$ CA rule such that:
\begin{itemize}
\item $F(0^L)=0^L$ and $F(1^L)=1^L$;
\item There is $m >1$ such that $0^L$ and $1^L$ are the only attractors of length $L$ of $F^m$.
\end{itemize}

Then $0^L$ and $1^L$ are the only attractors of length $L$ of $F$.
\end{prop}
\begin{proof}

By hypothesis, $0^L$ and $1^L$ are attractors of length $L$ of $F$. We will show that $F$ does not admit any other attractor.

Let $c$ be a configuration of length $L$ such that $c\notin \{0^L,1^L\}$.

Suppose $c$ is a fixed-point of $F$. Then we would have $F^m(c)=F^{m-1}(c)=\cdots =F(c)=c$ and $c$ would be a fixed-point of $F^m$, contradicting the fact that $0^L$ and $1^L$ are the only attractors of $F^m$.

Now, suppose $c$ belongs to a limit-cycle of period $p>1$ of $F$, that is, $F^p(c)=c$. Then:

\begin{equation*}
\left(\underbrace{F^m\circ F^m\circ\cdots\circ F^m}_{p\mbox{ times}}\right)(c)=F^{p\cdot m}(c)=F^{m\cdot p}(c)=\left(\underbrace{F^p\circ F^p\circ\cdots\circ F^p}_{m\mbox{ times}}\right)(c)=c
\end{equation*}

That is, $c$ belongs to a limit-cycle of period $p$ of $F^m$, contradicting the fact that $0^L$ and $1^L$ are the only attractors of $F^m$.

\end{proof}

\begin{corol}
\label{corol:x0xicomposition}

Let $f$ be a binary radius-$r$ CA rule such that $f^m(x_{-m\cdot\ceil{r}},\cdots,x_{m\cdot\floor{r}})=x_0\cdot x_i\cdot \phi(x_{[-m\cdot\ceil{r},m\cdot\floor{r}]})$, with $i\in\{-2,-1,1,2\}$ or such that $f^m(x_{-m\cdot\ceil{r}},\cdots,x_{m\cdot\floor{r}})=x_j\cdot x_{j+1}\cdot \phi(x_{[-m\cdot\ceil{r},m\cdot\floor{r}]}),-\ceil{r}\leq j\leq \floor{r}$, for some $m\geq 1$ and some function $\phi:\{0,1\}^{2mr+1}\rightarrow \{0,1\}$, such that:

\begin{itemize}
\item $f(0^{2r+1})=0$ and $f(1^{2r+1})=1$;
\item $\phi(1^{2mr+1})=1$;
\item $\phi(u_{mr})=0$, with $u_{mr}$ defined as in Corollary \ref{corol:x0x2}, if $i\in\{-2,2\}$.
\end{itemize}

Then $f$ is a consensus rule and $\mathscr{B}^L_1=\{1^L\}$ for all $L\geq 2r+1$.
\end{corol}

\begin{proof}

First notice that $F(0^L)=\left(f(0^{2r+1})\right)^L=0^L$ and $F(1^L)=\left(f(1^{2r+1})\right)^L=1^L$. Therefore, $F^m(0^L)=0^L$ and $F^m(1^L)=1^L$ also.

By applying Proposition \ref{prop:x0x1} (if $i\in\{-1,1\}$) or Corollary \ref{corol:x0x2} (if $i\in\{-2,2\}$) or Corollary \ref{corol:shifted} (in the remaining case) to $f^m$, we have that $0^L$ and $1^L$ are the only attractors of length $L$ of $F^m$ and $c$ is in the basin of attraction of $0^L$ for $f^m$ if, and only if, $c$ has any of its cells in state 0.

By Proposition \ref{prop:onlyattractors}, $0^L$ and $1^L$ are the only attractors of $F$, therefore $f$ is a consensus rule.

By the above, we already have that $\{1^L\}\subseteq \mathscr{B}^L_1$. Suppose $c\in\mathscr{C}_L$ is such that there is some integer $0\leq j\leq L-1$ such that $c_j=0$. Then there is an integer $k$ such that $(F^m)^k(c)=0^L$, therefore
\begin{equation*}
F^{m\cdot k}(c)=(F^m)^k(c)=0^L,
\end{equation*}

\noindent that is, $c\in\mathscr{B}^L_0$.

In other words, if $c\neq 1^L$, then $c\notin\mathscr{B}^L_1$ and we have $\mathscr{B}^L_1=\{1^L\}$.
\end{proof}

\subsubsection{Checking the properties for rules in Class A}

First, for every rule $f$ in class A we tested if there was a positive integer $m$ such that:

\begin{itemize}
\item $f^m$ is 0-forcing in positions $0$ and $i$, for some $i\in\{-2,-1,1,2\}$; and/or
\item $f^m$ is 0-forcing in positions $j$ and $j\pm 1$ or $j$ and $j\pm 2$, with $-\ceil{r}\leq j\leq \floor{r}$
\end{itemize}

Out of the 30,230 rules in class A, 27,251 rules satisfied at least one of the conditions above was satisfied for some $1\leq m \leq 5$. Greater values of $m$ were not tested due to the computational cost of computing $f^m$ (a radius-$(2m)$ rule).

Then for each $f^m$ we computed $\phi_{f^m}:\{0,1\}^{2mr+1}\rightarrow\{0,1\}^{2mr+1}$ such that

\begin{equation*}
f^m(x_{[-\ceil{mr},\floor{mr}]})=x_{i_1}\cdot x_{i_2}\cdot\phi_{f^m}(x_{[-\ceil{mr},\floor{mr}]})
\end{equation*}

\noindent with $i_1$ and $i_2$ being the positions in which $f^m$ is 0-forcing.

Finally, we checked whether or not the following conditions were satisfied:

\begin{itemize}
\item $\phi_{f^m}(1^{2mr+1})=1$, for every rule in class A;
\item $\phi_{f^m}(u_{mr})=0$, with $u_{mr}$ defined as in Corollary \ref{corol:x0x2}, for rules in class A with $(i_1,i_2)=(0,2)$ or $(i_1,i_2)=(0,-2)$.
\end{itemize}

 Notice that the 27,251 rules satisfying the conditions above also satisfy the conditions in Corollary \ref{corol:x0xicomposition}, hence such rules are consensus rules with $\mathscr{B}_1=\{1^L\}$.

Table \ref{tb:classA} shows how many rules in Class A are such that $f^m$ is 0-forcing, for $1\leq m\leq 5$, in two adjacent ($i$ and $i\pm 1$) or two alternated ($i$ and $i\pm 2$) positions. Out of the 30,230 rules, 27,251 rules satisfy such a property for some $1\leq m \leq 5$. Greater values of $m$ were not tested due to computational limitations.

\begin{table}
\caption{Number of rules in Class A such that $m$ is the smallest integer for which $f^m$ is 0-forcing in two adjacent ($i$ and $i\pm 1$) or two alternated ($i$ and $i\pm 2$) positions.\label{tb:classA}}
\begin{center}
\begin{tabular}{ c || c | c | c | c | c | c }
$m$ & 1 & 2 & 3 & 4 & 5 & Total \\
\hline
\# of rules & 324 & 5,284 & 11,108 & 7,292 & 3,243 & 27,251
\end{tabular}
\end{center}
\end{table}

\subsection{Class B}
\label{subsec:classB}

\subsubsection{Proofs for Class B}

The 14,680 rules in Class B behave as follows, for configurations of length $5\leq L\leq 20$:

\begin{itemize}
\item If $L$ is odd, $\mathscr{B}_1^L=\{1^L\}$ and $\mathscr{B}_0^L=\mathscr{C}_L-\{1^L\}$, that is, in such a case, the basins of attraction are exactly as the ones for the rules in Class A;
\item If $L$ is even, $\mathscr{B}_1^L=\{1^L,(10)^{L/2},(01)^{L/2}\}$ and $\mathscr{B}_0^L=\mathscr{C}_L-\{1^L,(10)^{L/2},(01)^{L/2}\}$.
\end{itemize}

In summary, rules in Class B detect whether or not a configuration of length $5\leq L\leq 20$ satisfies one of the following conditions: (a) all of its cells in state 1 or; (b) every pair of adjacent cells are in distinct states (that is, if it is a 01 or 10 necklace). This will be proven below, based upon some properties of such rules.

Let $f$ be a rule in Class B and suppose $f$ has the following properties, for some integer $n_1\geq 1$ and some integer $n_2\geq 1$:

\begin{itemize}
\item B1) For any $c\in\mathscr{C}_L$, such that $c\notin\{1^L\}$, if $L$ is odd, and $c\notin \{1^L,(01)^{L/2},(10)^{L/2}\}$, if $L$ is even, there is $0\leq j\leq L-1$ such that $\left(F^{n_1}(c)\right)_{[j,j+1]}=00$, that is, after iterating $F$ $n_1$ times over any configuration $c\notin \{1^L\}$, if $L$ is odd, and $c\notin\{1^L,(01)^{L/2},(10)^{L/2}\}$, if $L$ is even, the resulting configuration has two adjacent zeroes;

\item B2) There is $0\leq i\leq L-1$ such that if $c\in\mathscr{C}_L$ is such that $c_{[j,j+1]}=00$ for some $0\leq j\leq L-1$, then $\sigma_i\left(F^{n_2}(c)\right)_{[j,j+2]}=000$. In other words, every pair of adjacent zeroes in $c$ generates a triple of adjacent zeroes in $F^{n_2}(c)$, eventually displaced by $i$ cells from the original position of the pair of 0s in $c$.

\item B3) $f(0,1,0,1,0)=f(1,0,1,0,1)=f(1,1,1,1,1)=1$.

\item B4) $f(0,0,0,0,0)=0$
\end{itemize}

We will show that rules satisfying properties B1)--B4) have exactly the basins of attraction described in the beginning of this subsection.

\begin{prop}
\label{prop:growing00}
Assume $f$ is radius-$2$ CA rule satisfying property B2) for some integer $n\geq 1$ and let $c\in\mathscr{C}_L,L\geq 5$, be such that $c_{[l,l+k-1]}=0^k$ for some $0\leq l\leq L-1$ and some $2\leq k\leq L-1$. Then $\left(F^n(c)\right)$ has, at least, $k+1$ adjacent zeroes.
\end{prop}

\begin{proof}
By property B2), there is an integer $0\leq i\leq L-1$ such that if $c_{[j,j+1]}=00$, then $\sigma_i(F^n(c))_{[j,j+2]}=(F^n(c))_{[i+j,i+j+2]}=000$.

By hypothesis, $c$ satisfies such a condition for $j\in J=\{l,\cdots,l+k-2\}$, therefore $(F^n(c))_{[i+j,i+j+2]}=000$ for all $j\in J$, that is, $(F^n(c))_{[i+l,i+(l+k-2)+2]}=(F^n(c))_{[i+l,i+l+k]}=0^{k+1}$.

Notice that, since $k\leq L-1$, the indices $[i+l,i+l+k]$ all represent distinct cells in the configuration. Therefore $F^n(c)$ has, at least, $k+1$ adjacent 0s.
\end{proof}

\begin{corol}
\label{corol:growing00}

Let $f$ be a radius-$2$ CA rule satisfying properties B2) and B4) and $c\in\mathscr{C}_L$, $L\geq 5$, be such that it has two adjacent cells in state 0. Then there is an integer $m\geq 1$ such that $F^m(c)=0^L$. 
\end{corol}

\begin{proof}
If $c=0^L$, $F^1(c)=(f(0,0,0,0,0))^L=0^L$, by property B4). Assume from now on that $c\neq 0^L$. The remainder of the proof is inductive.

Suppose $c$ has exactly $2$ adjacent cells in state 0, we show that $F^{k\cdot n}(c)$ has $(k+2)$ cells in state $0$ (with a maximum of $L$) for all $k\geq 0$.
The base case $k=0$ is given by the hypothesis on $c$.

Assume $F^{k\cdot n}(c)\neq 0^L$ has $(k+2)$ adjacent cells in state $0$ for some integer $k\geq 0$. Applying Proposition \ref{prop:growing00} for $F^{k\cdot n}(c)$ implies that $F^n\left(F^{k\cdot n}(c)\right)=F^{(k+1)\cdot n}$ has at least $(k+3)$ adjacent cells in state 0.

It follows that for $k=L-2$ we have $F^{k\cdot n}(c)=0^L$, and the statement is obtained for $m=(L-2)\cdot n$.
\end{proof}

\begin{corol}
\label{corol:classB}
Let $f$ be a radius-$2$ CA rule satisfying properties B1) to B4). Then $f$ is a consensus rule and $\mathscr{B}_1^L=\{1^L\}$ if $L$ is odd and $\mathscr{B}_1^L=\{1^L,(10)^{L/2},(01)^{L/2}\}$ if $L$ is even.

\end{corol}
\begin{proof}

For any $L\geq 5$, property B3) implies that
$F(1^L)=(f(1,1,1,1,1))^L=1^L$, that is, $1^L\in \mathscr{B}_1^L$.

If $L$ is even, then:
\begin{eqnarray*}
F((10)^{L/2}) & = & (f(1,0,1,0,1)f(0,1,0,1,0))^{L/2}=(11)^{L/2}=1^L\mbox{ and }\\
F((01)^{L/2}) & = & (f(0,1,0,1,0)f(1,0,1,0,1))^{L/2}=(11)^{L/2}=1^L
\end{eqnarray*}

In summary, for odd $L$, $\{1^L\}\subseteq \mathscr{B}_1^L$ and for even $L$, $\{1^L,(01)^{L/2},(10)^{L/2}\}\subseteq \mathscr{B}_1^L$.

Now suppose $c\notin\{1^L\}$ if $L$ is odd and $c\notin \{1^L,(01)^{L/2},(10)^{L/2}\}$ if $L$ is even. Then, by property B1), there is an integer $n_1\geq 1$ such that $c'=F^{n_1}(c)$ has two adjacent 0s.

Finally, by applying Corollary \ref{corol:growing00}, there is an integer $m\geq 1$ such that $F^m(c')=F^{n_1+m}(c)=0^L$, implying that $c\in \mathscr{B}_0^L$, therefore $c\notin \mathscr{B}_1^L$.

Hence, $f$ is a consensus rule and $\mathscr{B}_1^L=\{1^L\}$ if $L$ is odd and $\mathscr{B}_1^L=\{1^L,(01)^{L/2},(10)^{L/2}\}$ if $L$ is even.
\end{proof}

\subsubsection{Checking the properties for rules in Class B}

Each one of the $14,680$ rules in class B satisfied properties $B1), B3)$ and $B4)$ and $12,294$ of such rules also satisfied property $B2)$ for some integer $1\leq n_2 \leq 5$. Greater values of $n_2$ were not tested due to computational limitations.

Such computational verifications were carried out as follows:

\begin{itemize}
\item \textbf{Property B1)}: First, for all $f$ in class B, we checked if there was an integer $m_1>0$ for which $f^{m_1}$ eventually preserves the existence of adjacent 0s in a configuration, that is, if $c\in\mathscr{C}_L$ contains $00$, then $F^{m_1}(c)$ also contains $00$.

Let $X_{(l,i)}=\{b=(b_1,\cdots,b_l)\in\{0,1\}^l:b_ib_{i+1}=00\}$ for integers $l,i$ with $1\leq i\leq l-1$ and $l\geq 2$. That is, $X_{(l,i)}$ is the set of all binary blocks of length $l$ containing $00$ as a sub-block in position $i$.

We then checked whether or not there were integers $m_1$ and $i_1$ such that the following condition was satisfied for every block $x=(x_1,\cdots,x_{4m_1+2})\in X_{(4m_1+2,i_1)}$:
\begin{equation*}
\left(f^{m_1}(x_1,\cdots,x_{4m_1+1}),f^{m_1}(x_2,\cdots,x_{4m_1+2})\right)=(0,0)
\end{equation*}

This shows that, if a configuration already has two adjacent 0s, then every configuration in its orbit under $f^{m_1}$ will also have two adjacent 0s, therefore it remained to take into account configurations not containing two adjacent 0s.

Suppose $c$ is a configuration of length $L$ such that $c$ does not have adjacent cells in state 0 and $c\notin\{1^L\}$, if $L$ is odd, and $c\notin\{1^L,(01)^{L/2},(10)^{L/2}\}$ if $L$ is even. Notice that, in such a case, $c$ must contain $1101$ and $1011$ as sub-blocks.

Let $Y_{(l,i)}=\{b=(b_1,\cdots,b_l)\in\{0,1\}^l:b_{[i,i+3]}\in\{1101,1011\}\}$ for integers $l,i$ with $1\leq i\leq l-3$ and $l\geq 4$. That is, $Y_{(l,i)}$ is the set of all binary blocks of length $l$ containing $1101$ or $1011$ as a sub-block in position $i$.

Then, for all $f$ in class B, we checked whether or not there were integers $m_2>0$ and $i_2$ such that, for every block $y=(y_1,\cdots,y_{4m_2+2})\in Y_{(4m_2+2,i_2)}$, the following condition was satisfied:
\begin{equation*}
\left(f^{m_2}(y_1,\cdots,y_{4m_2+1}),f^{m_2}(y_2,\cdots,y_{4m_2+2})\right)=(0,0)
\end{equation*}

This shows that, if a configuration $c$ contains $1101$ or $1011$ as a sub-block, then $F^{m_2}(c)$ contains two adjacent 0s.

By taking $n_1=max\{m_1,m_2\}$, we have that $F^{n_1}$ satisfies the conditions in property B1).

Table \ref{tb:classB1} shows the number of rules in class B satisfying such a condition for different values of $n_1$. Every rule in Class B satisfied condition B1) for some integer $n_1$.

\begin{table}
\caption{Number of rules in Class B such that $n_1$ is the smallest integer for which $f^{n_1}$ satisfies the condition B1).\label{tb:classB1}}
\begin{center}
\begin{tabular}{ c || c | c | c | c | c}
$n_1$ & 1 & 2 & 3 & 4 & Total\\
\hline
\# of rules & 11,848 & 2,793 & 0 & 39 & 14,680 
\end{tabular}
\end{center}
\end{table}

\item \textbf{Property B2)}:
In order to verify property B2), for every $f$ in class B we checked whether or not there were integers $n_2$ and $i_3$ such that, for every block $z=(z_1,\cdots,z_{4n_2+3})\in X_{(4n_2+3,i_3)}$ ,the following condition was satisfied:

\begin{equation*}
\left(f^{n_2}(z_1,\cdots,z_{4n_2+1}),f^{n_2}(z_2,\cdots,z_{4n_2+2}),f^{n_2}(z_3,\cdots,z_{4n_2+3})\right)=(0,0,0)
\end{equation*}

This is an extension of the preservation of adjacent 0s, and shows that, for every pair of adjacent cells in state 0 in a configuration, there is a triple of adjacent cells in state 0 in its image, possibly displaced by $i_3$ positions to the left or to the right (depending on $i_3$ being positive or negative) from the position of $00$ in the original configuration.

Table \ref{tb:classB2} shows the number of rules in class B satisfying such a condition for different values of $n_2$. Every rule in Class B satisfied such a condition for some integer $n_2$.

\begin{table}
\caption{Number of rules in Class B such that $n_2$ is the smallest integer for which $f^{n_2}$ satisfies the condition B2).\label{tb:classB2}}
\begin{center}
\begin{tabular}{ c || c | c | c | c | c | c }
$n_2$ & 1 & 2 & 3 & 4 & 5 & Total\\
\hline
\# of rules & 1,790 & 4,798 & 1,965 & 2,225 & 1,516 & 12,294
\end{tabular}
\end{center}
\end{table}

\item \textbf{Properties B3) and B4)}:
For every $f$ in class B, properties B3) and B4) were readily verified by checking the images of $(0,1,0,1,0),(1,0,1,0,1),(0,0,0,0,0)$ and $(1,1,1,1,1)$ under $f$.

\end{itemize}

Therefore, the above results show that (at least) 12,294 out of 14,680 rules in class B solve for all $L\geq 5$ the very same decision problem solved for $5 \leq L \leq 20$.

\subsection{Class C}
\label{subsec:classC}

\subsubsection{Proofs for Class C}

The proofs for rules in this class are very similar to the ones for class B, as they will rely on the expansion of blocks of three consecutive 0s rather than of blocks of two consecutive 0s.

The $789$ rules in Class C behave as follows, for configurations of length $5 \leq L\leq 20$:

\begin{itemize}
\item If $L$ is odd, $\mathscr{B}_1^L=\{1^L\}$ and $\mathscr{B}_0^L=\mathscr{C}_L-\{1^L\}$, that is, in such a case, the basins of attraction are exactly as the ones for the rules in Class A;

\item If $L=4p+2$ for some integer $p\geq 1$, $\mathscr{B}_1^L=\{1^L,(10)^{L/2},(01)^{L/2}\}$ and $\mathscr{B}_0^L=\mathscr{C}_L-\{1^L,(10)^{L/2},(01)^{L/2}\}$. That is, in such a case, the attractors are exactly as the ones for even $L$ for rules in Class B;

\item If $L=4p$ for some integer $p\geq 2$, $\mathscr{B}_1^L=\{1^L,(10)^{L/2},(01)^{L/2},(1100)^{L/4},(1001)^{L/4},(0011)^{L/4},(0110)^{L/4}\}$ and $\mathscr{B}_0^L=\mathscr{C}_L-\{1^L,(10)^{L/2},(01)^{L/2},(1100)^{L/4},(1001)^{L/4},(0011)^{L/4},(0110)^{L/4}\}$.

\end{itemize}

In summary, rules in Class C detect whether or not a configuration of length $L$ satisfies one of the following conditions: (a) all of its cells in state 1 or; (b) every pair of adjacent cells are in distinct states (that is, if it is a 01 or 10 necklace) or; (c) the configuration is a 1100 or 1001 or 0011 or 0110 necklace. This will be proven below, based upon some properties of such rules.

Let $f$ be a rule in Class C and assume $f$ has the following properties, for some integer $n_1\geq 1$ and some integer $n_2\geq 1$:

\begin{itemize}
\item C1) For any $c\in\mathscr{C}_L$, such that $c\notin\{1^L\}$, if $L$ is odd, $c\notin \{1^L,(01)^{L/2},(10)^{L/2}\}$, if $L=4p+2$ for some $p\geq 1$, and $c\notin \{1^L,(10)^{L/2},(01)^{L/2},(1100)^{L/4},(1001)^{L/4},(0011)^{L/4},(0110)^{L/4}\}$, if $L=4p$ for some $p\geq 2$, there is $0\leq j\leq L-1$ such that $\left(F^{n_1}(c)\right)_{[j,j+2]}=000$, that is, after iterating $F$ $n_1$ times over any configuration $c$ (as described above), the resulting configuration has three adjacent zeroes;

\item C2) There is $0\leq i\leq L-1$ such that if $c\in\mathscr{C}_L$ is such that $c_{[j,j+2]}=000$ for some $0\leq j\leq L-1$, then $\sigma_i\left(F^{n_2}(c)\right)_{[j,j+3]}=0000$. In other words, every triple of adjacent 0s in $c$ generates a quadruple of adjacent 0s in $F^{n_2}(c)$, eventually displaced by $i$ cells from the original position of the triple of 0s in $c$.

\item C3) $f(0,1,0,1,0)=f(1,0,1,0,1)=f(1,1,1,1,1)=f(1,1,0,0,1)=f(0,0,1,1,0)=1$.

\item C4) $f(0,0,0,0,0)=f(0,1,1,0,0)=f(1,0,0,1,1)=0$.
\end{itemize}

We will show that rules satisfying properties C1)--C4) have exactly the basins of attraction described in the beginning of this subsection.

\begin{prop}
\label{prop:growing000}

Assume $f$ is radius-$2$ CA rule satisfying property C2) for some integer $n\geq 1$ and let $c\in\mathscr{C}_L,L\geq 5$, be such that $c_{[l,l+k-1]}=0^k$ for some $0\leq l\leq L-1$ and some $3\leq k\leq L-1$. Then $\left(F^n(c)\right)$ has, at least, $k+1$ adjacent 0s.

\end{prop}

\begin{proof}

By property C2), there is an integer $0\leq i\leq L-1$ such that if $c_{[j,j+1]}=000$, then $\sigma_i(F^n(c))_{[j,j+3]}=(F^n(c))_{[i+j,i+j+3]}=0000$.

By hypothesis, $c$ satisfies such a condition for $j\in J=\{l,\cdots,l+k-2\}$, therefore $(F^n(c))_{[i+j,i+j+2]}=0000$ for all $j\in J$, that is, $(F^n(c))_{[i+l,i+(l+k-2)+2]}=(F^n(c))_{[i+l,i+l+k]}=0^{k+1}$.

Notice that, since $k\leq L-1$, the indices $[i+l,i+l+k]$ all represent distinct cells in the configuration. Therefore $F^n(c)$ has, at least, $k+1$ adjacent 0s.
\end{proof}

\begin{corol}
\label{corol:growing000}

Let $f$ be a radius-$2$ CA rule satisfying properties C2) and C4) and $c\in\mathscr{C}_L$, $L\geq 5$, be such that it has three adjacent cells in state 0. Then there is an integer $m\geq 1$ such that $F^m(c)=0^L$. 
\end{corol}

\begin{proof}
If $c=0^L$, $F^1(c)=(f(0,0,0,0,0))^L=0^L$, by property C4). Assume from now on that $c\neq 0^L$. The remainder of the proof is inductive.

Suppose $c$ has exactly $3$ adjacent cells in state 0, we show that $F^{k\cdot n}(c)$ has $(k+3)$ cells in state $0$ (with a maximum of $L$) for all $k\geq 0$.
The base case $k=0$ is given by the hypothesis on $c$.

Assume $F^{k\cdot n}(c)\neq 0^L$ has $(k+3)$ adjacent cells in state $0$ for some integer $k\geq 0$. Applying Proposition \ref{prop:growing000} for $F^{k\cdot n}(c)$ implies that $F^n\left(F^{k\cdot n}(c)\right)=F^{(k+1)\cdot n}$ has at least $(k+4)$ adjacent cells in state 0.

It follows that for $k=L-3$ we have $F^{k\cdot n}(c)=0^L$, and the statement is obtained for $m=(L-3)\cdot n$.
\end{proof}

\begin{corol}
\label{corol:classC}
Let $f$ be a radius-$2$ CA rule satisfying properties C1) to C4). Then $f$ is a consensus rule and 

\begin{itemize}
\item $\mathscr{B}_1^L=\{1^L\}$ if $L$ is odd;
\item $\mathscr{B}_1^L=\{1^L,(10)^{L/2},(01)^{L/2}\}$ if $L=4p+2$ for some integer $p\geq 1$; 
\item $\mathscr{B}_1^L=\{1^L,(10)^{L/2},(01)^{L/2},(1100)^{L/4},(1001)^{L/4},(0011)^{L/4},(0110)^{L/4}\}$ if $L=4p$ for some integer $p\geq 2$.
\end{itemize}
\end{corol}

\begin{proof}
For any $L\geq 5$, property C3) implies that
$F(1^L)=(f(1,1,1,1,1))^L=1^L$, that is, $1^L\in \mathscr{B}_1^L$.

If $L=4p+2$ for some integer $p\geq 1$, then:
\begin{eqnarray*}
F((10)^{L/2}) & = & (f(1,0,1,0,1)f(0,1,0,1,0))^{L/2}=(11)^{L/2}=1^L\mbox{ and }\\
F((01)^{L/2}) & = & (f(0,1,0,1,0)f(1,0,1,0,1))^{L/2}=(11)^{L/2}=1^L
\end{eqnarray*}

Also, if $L=4p$ for some integer $p\geq 2$, then:
\begin{eqnarray*}
F^2((1100)^{L/4}) &=&F\left(\left(f(0,0,1,1,0)f(0,1,1,0,0)f(1,1,0,0,1)f(1,0,0,1,1)\right)^{L/4}\right)= \\ 
&=& F((10)^{L/2}) =(f(1,0,1,0,1)f(0,1,0,1,0))^{L/2}=(11)^{L/2}=1^L
\end{eqnarray*}

\begin{eqnarray*}
F^2((1001)^{L/4}) &=&F\left(\left(f(0,1,1,0,0)f(1,1,0,0,1)f(1,0,0,1,1)f(0,0,1,1,0)\right)^{L/4}\right)= \\ 
&=& F((01)^{L/2}) =(f(0,1,0,1,0)f(1,0,1,0,1))^{L/2}=(11)^{L/2}=1^L
\end{eqnarray*}

\begin{eqnarray*}
F^2((0011)^{L/4}) &=&F\left(\left(f(1,1,0,0,1)f(1,0,0,1,1)f(0,0,1,1,0)f(0,1,1,0,0)\right)^{L/4}\right)= \\ 
&=& F((10)^{L/2}) =(f(1,0,1,0,1)f(0,1,0,1,0))^{L/2}=(11)^{L/2}=1^L
\end{eqnarray*}

\begin{eqnarray*}
F^2((0110)^{L/4}) &=&F\left(\left(f(1,0,0,1,1)f(0,0,1,1,0)f(0,1,1,0,0)f(1,1,0,0,1)\right)^{L/4}\right)= \\ 
&=& F((01)^{L/2}) =(f(0,1,0,1,0)f(1,0,1,0,1))^{L/2}=(11)^{L/2}=1^L
\end{eqnarray*}

In summary:
\begin{itemize}
\item If $L$ is odd, then $\{1^L\}\subseteq \mathscr{B}_1^L$;
\item If $L=4p+2$, $p\geq 1$, then $\{1^L,(01)^{L/2},(10)^{L/2}\}\subseteq \mathscr{B}_1^L$;
\item If $L=4p$, $p\geq 2$, then $\{1^L,(10)^{L/2},(01)^{L/2},(1100)^{L/4},(1001)^{L/4},(0011)^{L/4},(0110)^{L/4}\}\subseteq \mathscr{B}_1^L$
\end{itemize}

Now suppose $c\notin\{1^L\}$ if $L$ is odd, $c\notin \{1^L,(01)^{L/2},(10)^{L/2}\}$ if $L=4p+2$ for some integer $p\geq 1$ and $c\notin \{1^L,(10)^{L/2},(01)^{L/2},(1100)^{L/4},(1001)^{L/4},(0011)^{L/4},(0110)^{L/4}\}$ if $L=4p$ for some integer $p\geq 2$. Then, by property C1), there is an integer $n_1\geq 1$ such that $c'=F^{n_1}(c)$ has three adjacent 0s.

Finally, by applying Corollary \ref{corol:growing000}, there is an integer $m\geq 1$ such that $F^m(c')=F^{n_1+m}(c)=0^L$, implying that $c\in \mathscr{B}_0^L$, therefore $c\notin \mathscr{B}_1^L$.

Hence, $f$ is a consensus rule such that: 
\begin{itemize}
\item $\mathscr{B}_1^L=\{1^L\}$ if $L$ is odd;
\item $\mathscr{B}_1^L=\{1^L,(01)^{L/2},(10)^{L/2}\}$ if $L=4p+2$ for some integer $p\geq 1$;
\item $\mathscr{B}_1^L=\{1^L,(10)^{L/2},(01)^{L/2},(1100)^{L/4},(1001)^{L/4},(0011)^{L/4},(0110)^{L/4}\}$ if $L=4p$ for some integer $p\geq 2$.
\end{itemize}
\end{proof}

\subsubsection{Checking the properties for rules in Class C}
Each one of the $789$ rules in class C satisfied properties $C3)$ and $C4)$, $783$ satisfied property $C1)$ for some $1\leq n_1 \leq 9$ and $714$ satisfied property $C2)$ for some integer $1\leq n_2 \leq 5$. Greater values of $n_1$ and $n_2$ were not tested due to computational limitations.

Such computational verifications were carried out as follows:

\begin{itemize}
\item \textbf{Property C1)}: Let $n$ be a positive integer. A block $b=(b_1,\cdots,b_{4n+3})$ is said to be a \emph{failing block of order} $n$ when:
\begin{equation*}
f^n(b_{[1,4n+1]})f^n(b_{[2,4n+2]})f^n(b_{[4n+3]})\neq 000
\end{equation*}

In other words, a block of length $4n+3$ is a failing block of order $n$ when its image under $f^n$ is not a triple of 0s.

Let $\mathscr{F}_n=\{b\in\{0,1\}^{4n+3}:b\mbox{ is a failing block of order }n\}$ be the \emph{set of failing blocks of order} $n$.

Now, consider the directed graph $\mathcal{G}_n$ whose set of vertices is $\mathscr{F}_n$ and the set of edges is given by:
\begin{equation*}
\mathcal{E}_n=\{(p_{[1,4n+3]},q_{[1,4n+3]})\in\mathscr{F}_n\times\mathscr{F}_n:p_{[2,4n+3]}=q_{[1,4n+2]}\}
\end{equation*}

That is, the vertices in $\mathcal{G}_n$ are the failing blocks of order $n$ and two such blocks are connected when the last $(4n+2)$ bits of the first failing block coincide with the first $(4n+2)$ bits of the last failing block.

Notice that cycles in $\mathscr{G}_n$ represent (periodic) configurations consisting exclusively of failing blocks of order $n$. That is, such configurations, when updated under $f^n$, do not contain three adjacent 0s.

On the other hand, any configuration that cannot be described by a cycle in such a graph must present three adjacent 0s in its image under $f^n$, since it must contain at least one block not in $\mathscr{F}_n$.

For any $n\geq 1$, $\mathcal{G}_n$ always contains (at least) the cycles describing the following configurations, which are in the basin of attraction of $1^L$,$L\geq 5$:

\begin{itemize}
\item $1^L$, for any $L\geq 5$;
\item $(10)^{L/2}$ and $(01)^{L/2}$, for even $L$
\item $(1100)^{L/4}$, $(1001)^{L/4}$, $(0011)^{L/4}$ and $(0110)^{L/4}$, for $L=4k$, for some integer $k\geq 2$.
\end{itemize}

Such cycles will be referred to as the \emph{trivial failing cycles of }$f^n$.

For $783$ out of the $789$ rules in class $C$, there is $1\leq n_1\leq 9$ such that the only cycles in $\mathcal{G}_{n_1}$ are the trivial failing cycles of $f^{n_1}$, therefore, for any configuration apart from the ones described above, there is at least one block of three adjacent 0s in its image under $f^{n_1}$.

Table \ref{tb:classC1} shows the number of rules in class C satisfying property C1) for different values of $n_1$.

\begin{table}
\caption{Number of rules in Class C such that $n_1$ is the smallest integer for which $f^{n_1}$ satisfies the property C1).\label{tb:classC1}}
\begin{center}
\begin{tabular}{ c || c | c | c | c | c | c | c | c | c | c }
$n_1$ & 1 & 2 & 3 & 4 & 5 & 6 & 7 & 8 & 9 & Total\\
\hline
\# of rules & 32 & 263 & 331 & 110 & 33 & 10 & 2 & 0 & 2 & 783
\end{tabular}
\end{center}
\end{table}

\item \textbf{Property C2)}: Property C2) was achieved for $714$ out of the $789$ rules in class C for some $1\leq n_2 \leq 5$.

Let $W_{(l,i)}=\{b=(b_1,\cdots,b_l)\in\{0,1\}^l:b_ib_{i+1}b_{i+2}=000\}$ for integers $l,i$ with $1\leq i\leq l-2$ and $l\geq 3$. That is, $W_{(l,i)}$ is the set of all binary blocks of length $l$ containing $000$ as a sub-block in position $i$.

For every $f$ in class C we checked whether or not there were integers $n_2$ and $i_3$ such that, for every block $z=(z_1,\cdots,z_{4n_2+4})\in X_{(4n_2+4,i_3)}$, the following condition was satisfied:
\begin{equation*}
\left(f^{n_2}(z_1,\cdots,z_{4n_2+4}),f^{n_2}(z_2,\cdots,z_{4n_2+2}),f^{n_2}(z_3,\cdots,z_{4n_2+3}),f^{n_2}(z_3,\cdots,z_{4n_2+3})\right)=(0,0,0,0)
\end{equation*}

This is analogous to condition $B2)$ (expansion of 00s) in class B, and shows that, for every triple of adjacent cells in state 0 in a configuration, there is a quadruple of adjacent cells in state 0 in its image under $f^{n_2}$, possibly displaced by $i_3$ positions to the left or to the right (depending on $i_3$ being positive or negative) from the position of $000$ in the original configuration.

Table~\ref{tb:classC2} shows the number of rules in class C satisfying such a condition for different values of $n_2$.

\begin{table}
\caption{Number of rules in Class C such that $n_2$ is the smallest integer for which $f^{n_2}$ satisfies the condition C2).\label{tb:classC2}}
\begin{center}
\begin{tabular}{ c || c | c | c | c | c | c }
$n_2$ & 1 & 2 & 3 & 4 & 5 & Total\\
\hline
\# of rules & 468 & 114 & 26 & 61 & 45 & 714
\end{tabular}
\end{center}
\end{table}

\item \textbf{Properties C3) and C4)}: For every $f$ in class C, properties C3) and C4) were readily verified by checking the images of $(0,0,1,1,0), (0,1,1,0,0),(1,1,0,0,1), (1,0,0,1,1), (0,1,0,1,0),(1,0,1,0,1),(0,0,0,0,0)$ and $(1,1,1,1,1)$ under $f$.

\end{itemize}

Out of the $789$ rules in class C, $80$ do not satisfy $C1)$ or $C2)$ for the tested values, therefore, the results above show that at least $709$ rules in $C$ solve the same decision problem for all $L\geq 5$ that they solved for $5 \leq L \leq 20$.

\section{Concluding remarks}
\label{sec:conclusion}

\subsection{Retrospective and prospects}
\label{subsec:summing_up}

Achieving solutions of decision problems by distributed consensus in cellular automata with finite configurations is a challenging endeavour because they have to be valid for any configuration size. However, this research direction is clearly motivating insofar as it can lead to a way of solving a global problem by means of a totally local fashion, in consensus with all the parts concerned.

In tune with that, after a reappraisal of the related results available in the literature for binary one-dimensional CAs with radius 2, we significantly extended what was known about the possible decision problems in the space, by settling the matter on many more rules and at the same time identifying a number of other candidate decision problems. More precisely, a single problem had been characterised (the one corresponding to our Class A) and associated to only 1,121 rules, a number that increased here to 27,251. The other two problems we analysed, related to Classes B and C, led to settle 12,294 and 709 rules for each class, respectively. So, out of the 54,928 radius-$2$ consensus rules obtained for $5 \leq L \leq 20$, at least 40,254 rules (representing, approximately, $73.3 \%$ of the total amount of potential consensus rules) really do solve decision problems for configurations of any length $L \geq 5$. 

As for the 9,229 uncharacterised rules, recall that, as mentioned earlier, while some of them may turn out not to be decision problems at all, as they might be filtered out for configuration sizes larger than $L=20$, other rules remain as true candidates, related to additional decision problems not analysed here, or even to the decision problems of the Classes A, B or C, which we could not establish due to a possible partiality of our characterisation. 

It is also worth remarking that the problems related to Classes A, B and C, as well as those that might be eventually associated to the remaining rules, have (or would have, for the latter set) a description that is biased by the rule we currently take as reference in its dynamical equivalence class, namely, the rule with the smallest number in the class (which is made up by the rule itself and the related symmetrical rules due to negation, reflection or both). This means, for instance for the rules in Class A, that the corresponding negated rules would have $1^L$ as the main fixed point (the one with almost all configurations converging to it), not to $0^L$, as in Class A. Nevertheless, these other problems are essentially the same up to negation, reflection or both.

In contrast to the standard synchronous/parallel update, deterministic block sequential updates refer to the form of asynchronism achieved by partitioning a configuration and assigning an update priority to each of its blocks, from 1 up to the number of blocks \cite{AracenaMarco2011}; in this sense, the synchronous update is simply the special case where the partition has a single block with priority 1. It is known that any fixed point is invariant to any block sequential update, which means that any existing fixed point in the synchronous case should also exist in any other asynchronous update. Consequently, if a rule is to be truly thought to be solving a decision problem by distributed consensus, its attractors can only be the two fixed points $0^L$ and $1^L$, regardless of the block sequential update employed. So, by subjecting the action of the remaining 14,674 rules to all possible deterministic block sequential updates from increasing sizes $L$, a further filtering process can be carried out towards likelier consensus type rules, since any rule displaying a new attractor for some asynchronous update cannot be granted the consensus type status and, therefore, can be discarded. This is a sensible supplemental procedure to the filtering done here, as this would help reinforce the status of consensus type rule to the remaining candidates. 

Notwithstanding the difficulty of formally charactering the dynamics of the rules involved, the actual decision problems concerned are still very simple, as showed. So, consensus type rules solving more elaborate and interesting problems must be searched among the remaining 14,674 radius-$2$ rules, as discussed below and in the Appendix. And characterising them undoubtedly constitutes the natural follow-up to the present work. 

\subsection{A view on all candidate decision problems of the space}
\label{subsec:discussion}

In the Appendix we provide a complete listing of the 54,928 rules that were obtained out of the filtering process carried out to identify candidate consensus rules in the CA space we analysed. Each piece of data refers to a pattern containing the number of configurations associated to the fixed point $1^L$, for every size $L$ in ascending order from $L=5$ to $L=20$, and the amount of rules displaying the corresponding pattern. In total we identified 485 distinct patterns, which are shown in lexicographic order.

For instance, with respect to the decision problems we characterised here, the pattern \{1, 1, 1, 1, 1, 1, 1, 1, 1, 1, 1, 1, 1, 1, 1, 1\} corresponds to all 30,230 rules in Class A and \{1, 3, 1, 3, 1, 3, 1, 3, 1, 3, 1, 3, 1, 3, 1, 3\} refers to all 14,680 rules in Class B. As for the rules in Class C, the pattern \{1, 3, 1, 7, 1, 3, 1, 7, 1, 3, 1, 7, 1, 3, 1, 7\} refers to 1,223 rules, which includes all 709 characterised Class C rules, but also others that, as mentioned earlier, might still be solving the Class C decision problem but our characterisation was insufficient to account for them; rules that might be solving different decision problems; or others that would not be solving a decision problem at all, but were not filtered out within the configuration sizes the filtering process relied upon. By the way, these three possibilities apply to all 14,674 rules not yet accounted for. 

Among them, some rules may be quite likely very easy to characterise because they exhibit a straightforward pattern, as is the case for the 9 rules associated to pattern \{31, 63, 127, 255, 511, 1023, 2047, 4095, 8191, 16383, 32767, 65535, 131071, 262143, 524287, 1048575\}, which clearly refers to basins of attraction with $2^L-1$ configurations leading to the fixed point point $1^L$, analogously to Class A. Equally easy seems the characterisation of the 23 rules linked to the pattern \{31, 61, 127, 253, 511, 1021, 2047, 4093, 8191, 16381, 32767, 65533, 131071, 262141, 524287, 1048573\}, which additionally imply that the configurations $(01)^{L/2}$ and $(10)^{L/2}$ converge to $1^L$, in all even-sized configurations, analogously to what happens in Class B. By inspecting the list of patterns, various others also seem to be subject of an easy characterisation, such as the 342 rules associated to \{1, 3, 1, 11, 1, 3, 1, 11, 1, 3, 1, 11, 1, 3, 1, 11\}, the 236 rules associated to \{1, 9, 1, 3, 1, 3, 1, 9, 1, 3, 1, 3, 1, 9, 1, 3\}, etc.

But while the simple characterisations are suggestive of overly simple related decision problems, many other rules might be leading to possibly more interesting and relevant problems, even though requiring more sophisticated analyses. This might be the case of 258 candidate problems which are associated to a single rule each. Or, more precisely, the case, for instance, of the 30 rules corresponding to the pattern \{21, 39, 71, 131, 241, 443, 815, 1499, 2757, 5071, 9327, 17155, 31553, 58035, 106743, 196331\}, which seems to comply with the integer sequence A001644 \cite{A001644}, which for $L\geq 3$, $a(L)$ is the number of cyclic sequences consisting of $L$ 0s and 1s that do not contain three consecutive 1s, provided the positions of the 0s and 1s are fixed on a circle. Many surprises such as this is quite likely hidden in the list of patterns, waiting to be uncovered.

\section*{Acknowledgements}

This work has been partially funded by the HORIZON-MSCA-2022-SE-01 project 101131549 ``Application-driven Challenges for Automata Networks and Complex Systems (ACANCOS)'', and by the STIC-AmSud 22-STIC-02 project ``Consensus Cellular Automata Algorithms and Tie-Majority Classification'' (CAMA) provided by CAPES 88881.694458/2022-01, ANID and MEAE). P.P.B. additionally thanks CNPq for the research grant PQ 303356/2022-7, and CAPES for Mackenzie-PrInt research grant no.\ 88887.310281/2018-00. K.P. additionally thanks project ANR-24-CE48-7504 ALARICE. E.G. and M.M-M. additionally thanks project FONDECYT/ANID Regular 1250984.

\bibliographystyle{plainnat}
\bibliography{references}

\section*{\center{Appendix: Listing of all candidate or characterised decision problems}}
\label{appendix}

Here we provide a listing of all 54,928 rules that were obtained out of the filtering process carried out to identify candidate consensus rules in the CA space we analysed. Each piece of data refers to a pattern containing the number of configurations associated to the fixed point $1^L$, for every size $L$ in ascending order from $L=5$ to $L=20$, and the amount of rules displaying the corresponding pattern. In total we identified 485 distinct patterns, shown here in lexicographic order, for easier reference. The corresponding rules in each of them are omitted here, but are all available upon request.

\begin{longtable}{lr}
\textbf{Pattern} & \small{\textbf{Frequency}} \\\hline
\endhead
 \small{(1,1,1,1,1,1,1,1,1,1,1,1,1,1,1,1)} & \small{30230} \\
 \small{(1,1,1,5,1,1,1,5,1,1,1,5,1,1,1,5)} & \small{19} \\
 \small{(1,3,1,3,1,3,1,3,1,3,1,3,1,3,1,3)} & \small{14680} \\
 \small{(1,3,1,3,1,13,1,15,1,3,1,3,1,3,1,13)} & \small{1} \\
 \small{(1,3,1,7,1,3,1,7,1,3,1,7,1,3,1,7)} & \small{1223} \\
 \small{(1,3,1,7,1,3,1,19,1,3,1,7,1,3,1,7)} & \small{8} \\
 \small{(1,3,1,11,1,3,1,3,1,3,1,11,1,3,1,3)} & \small{13} \\
 \small{(1,3,1,11,1,3,1,11,1,3,1,11,1,3,1,11)} & \small{342} \\
 \small{(1,3,1,11,1,13,1,3,1,3,1,11,1,21,1,13)} & \small{45} \\
 \small{(1,3,1,11,1,13,1,23,1,17,1,27,1,21,1,41)} & \small{1} \\
 \small{(1,3,1,11,1,23,1,15,1,3,1,11,1,39,1,63)} & \small{18} \\
 \small{(1,3,1,15,1,3,1,7,1,3,1,15,1,3,1,7)} & \small{55} \\
 \small{(1,3,1,15,1,3,1,19,1,3,1,15,1,3,1,27)} & \small{2} \\
 \small{(1,3,1,15,1,3,1,19,1,3,1,31,1,3,1,47)} & \small{1} \\
 \small{(1,3,1,15,1,13,1,19,1,17,1,31,1,39,1,57)} & \small{7} \\
 \small{(1,3,1,15,1,23,1,31,1,31,1,47,1,75,1,127)} & \small{4} \\
 \small{(1,3,1,19,1,3,1,3,1,3,1,19,1,3,1,3)} & \small{8} \\
 \small{(1,3,1,19,1,3,1,23,1,3,1,35,1,3,1,51)} & \small{4} \\
 \small{(1,3,1,19,1,3,1,35,1,3,1,67,1,3,1,131)} & \small{26} \\
 \small{(1,3,1,23,1,3,1,7,1,3,1,23,1,3,1,7)} & \small{2} \\
 \small{(1,3,1,23,1,3,1,19,1,3,1,39,1,3,1,47)} & \small{1} \\
 \small{(1,3,1,27,1,3,1,35,1,3,1,59,1,3,1,91)} & \small{1} \\
 \small{(1,3,1,27,1,3,1,35,1,3,1,75,1,3,1,131)} & \small{6} \\
 \small{(1,3,1,35,1,3,1,35,1,3,1,83,1,3,1,131)} & \small{2} \\
 \small{(1,4,1,1,4,1,1,4,1,1,4,1,1,4,1,1)} & \small{3205} \\
 \small{(1,4,1,1,22,1,1,4,1,1,4,1,1,22,1,1)} & \small{4} \\
 \small{(1,4,8,21,4,11,12,16,14,22,19,37,35,40,39,63)} & \small{1} \\
 \small{(1,4,15,37,4,11,12,16,14,29,34,53,35,40,39,63)} & \small{1} \\
 \small{(1,6,1,3,4,3,1,6,1,3,4,3,1,6,1,3)} & \small{926} \\
 \small{(1,6,1,7,4,3,1,10,1,3,4,7,1,6,1,7)} & \small{9} \\
 \small{(1,6,1,11,4,3,1,6,1,3,4,11,1,6,1,3)} & \small{1} \\
 \small{(1,6,1,11,4,3,1,14,1,3,4,11,1,6,1,11)} & \small{17} \\
 \small{(1,6,1,15,4,3,1,10,1,3,4,15,1,6,1,7)} & \small{2} \\
 \small{(1,7,1,1,7,1,1,7,1,1,7,1,1,7,1,1)} & \small{369} \\
 \small{(1,7,1,1,16,1,1,7,1,1,7,1,1,16,1,1)} & \small{3} \\
 \small{(1,7,1,9,7,1,23,7,1,43,7,9,69,7,39,101)} & \small{2} \\
 \small{(1,7,1,17,7,1,23,7,1,43,7,17,69,7,39,101)} & \small{2} \\
 \small{(1,7,15,49,118,221,507,1075,2315,4971,10531,22241,46700,97648,203643,423681)} & \small{1} \\
 \small{(1,7,15,49,118,241,518,1099,2341,5027,10696,22529,47244,98782,206037,428641)} & \small{1} \\
 \small{(1,7,15,49,118,241,518,1111,2354,5055,10756,22641,47465,99232,207006,430601)} & \small{1} \\
 \small{(1,7,22,49,100,211,496,1051,2289,4838,10261,21633,45374,95056,198323,412751)} & \small{1} \\
 \small{(1,7,22,49,100,211,496,1063,2289,4880,10306,21745,45595,95542,199330,414751)} & \small{1} \\
 \small{(1,7,22,57,118,241,540,1159,2523,5356,11311,23817,49964,104488,217760,452581)} & \small{1} \\
 \small{(1,7,22,57,118,241,540,1171,2523,5356,11326,23849,50032,104614,218026,453121)} & \small{1} \\
 \small{(1,7,22,57,118,251,551,1195,2562,5440,11491,24185,50729,106036,220933,459071)} & \small{4} \\
 \small{(1,7,57,129,163,351,782,1675,3550,7491,15661,32657,68392,142189,294653,608591)} & \small{1} \\
 \small{(1,9,1,3,1,3,1,9,1,3,1,3,1,9,1,3)} & \small{236} \\
 \small{(1,9,1,3,1,3,1,21,1,3,1,3,1,9,1,3)} & \small{1} \\
 \small{(1,9,1,3,7,3,1,9,1,3,7,3,1,9,1,3)} & \small{196} \\
 \small{(1,9,1,7,1,13,1,13,1,17,1,23,1,27,1,37)} & \small{96} \\
 \small{(1,9,1,7,7,3,1,13,1,3,7,7,1,9,1,7)} & \small{2} \\
 \small{(1,9,1,11,1,3,1,9,1,17,1,11,1,9,1,23)} & \small{6} \\
 \small{(1,9,1,11,1,13,1,9,1,17,1,27,1,27,1,33)} & \small{8} \\
 \small{(1,9,1,11,1,13,1,17,1,17,1,27,1,27,1,41)} & \small{12} \\
 \small{(1,9,1,11,1,23,1,21,1,17,1,43,1,63,1,83)} & \small{4} \\
 \small{(1,9,1,15,1,23,1,25,1,45,1,63,1,99,1,147)} & \small{4} \\
 \small{(1,9,1,15,1,33,1,37,1,59,1,95,1,153,1,257)} & \small{2} \\
 \small{(1,9,1,19,1,3,1,9,1,17,1,19,1,9,1,23)} & \small{5} \\
 \small{(1,9,1,19,1,13,1,29,1,31,1,51,1,63,1,101)} & \small{1} \\
 \small{(1,9,1,19,1,23,1,41,1,59,1,99,1,153,1,251)} & \small{4} \\
 \small{(1,9,1,19,7,3,1,17,1,3,7,19,1,9,1,11)} & \small{2} \\
 \small{(1,9,1,23,1,13,1,25,1,31,1,55,1,63,1,97)} & \small{1} \\
 \small{(1,9,1,27,1,13,1,41,1,31,1,75,1,63,1,141)} & \small{1} \\
 \small{(1,9,15,51,136,283,617,1317,2796,5981,12661,26451,55421,115686,240788,499583)} & \small{1} \\
 \small{(1,9,22,59,136,293,650,1377,2926,6212,13081,27499,57512,119916,249281,516933)} & \small{1} \\
 \small{(1,9,22,59,136,293,650,1389,2952,6254,13186,27707,57971,120852,251219,520873)} & \small{1} \\
 \small{(1,9,22,67,127,293,639,1353,2913,6198,13096,27491,57478,119925,249376,517073)} & \small{1} \\
 \small{(1,10,1,1,4,1,1,10,1,1,4,1,1,10,1,1)} & \small{69} \\
 \small{(1,12,1,3,4,3,1,12,1,3,4,3,1,12,1,3)} & \small{62} \\
 \small{(1,12,1,7,4,13,1,16,1,17,4,23,1,30,1,37)} & \small{8} \\
 \small{(1,12,1,11,4,13,1,20,1,17,4,27,1,30,1,41)} & \small{3} \\
 \small{(1,13,1,1,7,1,1,13,1,1,7,1,1,13,1,1)} & \small{195} \\
 \small{(1,13,1,1,7,1,1,25,1,1,7,1,1,13,1,1)} & \small{1} \\
 \small{(1,13,1,1,16,1,1,13,1,1,22,1,1,22,1,1)} & \small{2} \\
 \small{(1,13,1,1,16,1,1,25,1,1,37,1,1,58,1,1)} & \small{220} \\
 \small{(1,13,1,9,7,1,23,13,1,43,7,9,69,13,39,101)} & \small{2} \\
 \small{(1,13,1,17,7,1,23,13,1,43,7,17,69,13,39,101)} & \small{1} \\
 \small{(1,15,1,3,1,3,1,15,1,3,1,3,1,15,1,3)} & \small{33} \\
 \small{(1,15,1,3,7,3,1,15,1,3,7,3,1,15,1,3)} & \small{31} \\
 \small{(1,15,1,3,7,3,1,15,1,3,7,3,1,33,1,3)} & \small{1} \\
 \small{(1,15,1,3,7,3,1,27,1,3,7,3,1,15,1,3)} & \small{3} \\
 \small{(1,15,1,3,16,3,1,27,1,3,37,3,1,60,1,3)} & \small{39} \\
 \small{(1,15,1,7,1,13,1,19,1,17,1,23,1,33,1,37)} & \small{6} \\
 \small{(1,15,1,7,1,23,1,31,1,31,1,71,1,87,1,127)} & \small{2} \\
 \small{(1,15,1,7,16,3,1,31,1,3,37,7,1,60,1,7)} & \small{1} \\
 \small{(1,15,1,11,1,13,1,15,1,17,1,27,1,33,1,33)} & \small{2} \\
 \small{(1,15,1,11,1,13,1,27,1,31,1,43,1,69,1,93)} & \small{2} \\
 \small{(1,15,1,11,1,23,1,23,1,31,1,43,1,51,1,71)} & \small{6} \\
 \small{(1,15,1,15,1,23,1,43,1,59,1,95,1,159,1,247)} & \small{1} \\
 \small{(1,15,1,15,1,33,1,43,1,73,1,127,1,195,1,337)} & \small{1} \\
 \small{(1,16,1,1,4,1,1,16,1,1,4,1,1,16,1,1)} & \small{16} \\
 \small{(1,16,1,1,4,1,1,28,1,1,4,1,1,16,1,1)} & \small{10} \\
 \small{(1,16,1,1,4,1,1,28,1,1,4,1,1,34,1,1)} & \small{1} \\
 \small{(1,16,1,1,4,1,1,28,1,1,4,1,1,52,1,1)} & \small{3} \\
 \small{(1,16,1,1,4,1,1,40,1,1,4,1,1,16,1,1)} & \small{1} \\
 \small{(1,16,1,1,4,1,1,40,1,1,4,1,1,34,1,1)} & \small{1} \\
 \small{(1,16,1,1,4,1,1,88,1,1,4,1,1,52,1,1)} & \small{1} \\
 \small{(1,18,1,3,4,3,1,18,1,3,4,3,1,18,1,3)} & \small{20} \\
 \small{(1,18,1,3,4,3,1,30,1,3,4,3,1,54,1,3)} & \small{2} \\
 \small{(1,18,1,7,4,23,1,34,1,31,4,71,1,90,1,127)} & \small{1} \\
 \small{(1,18,1,11,4,23,1,26,1,31,4,43,1,54,1,71)} & \small{4} \\
 \small{(1,19,1,1,7,1,1,19,1,1,7,1,1,19,1,1)} & \small{69} \\
 \small{(1,19,1,1,7,1,1,31,1,1,7,1,1,19,1,1)} & \small{5} \\
 \small{(1,19,1,1,7,1,1,31,1,1,7,1,1,37,1,1)} & \small{1} \\
 \small{(1,19,1,1,7,1,1,31,1,1,7,1,1,55,1,1)} & \small{11} \\
 \small{(1,19,1,1,7,1,1,43,1,1,7,1,1,55,1,1)} & \small{1} \\
 \small{(1,19,1,1,7,1,1,55,1,1,7,1,1,37,1,1)} & \small{1} \\
 \small{(1,19,1,1,7,1,1,67,1,1,7,1,1,37,1,1)} & \small{1} \\
 \small{(1,19,1,1,16,1,1,19,1,1,7,1,1,28,1,1)} & \small{2} \\
 \small{(1,19,1,1,16,1,1,19,1,1,37,1,1,28,1,1)} & \small{1} \\
 \small{(1,19,1,1,16,1,1,31,1,1,37,1,1,64,1,1)} & \small{35} \\
 \small{(1,19,1,1,16,1,1,43,1,1,52,1,1,118,1,1)} & \small{10} \\
 \small{(1,19,1,1,16,1,1,67,1,1,52,1,1,118,1,1)} & \small{1} \\
 \small{(1,19,1,1,25,1,1,31,1,1,52,1,1,73,1,1)} & \small{2} \\
 \small{(1,19,1,1,25,1,1,31,1,1,67,1,1,91,1,1)} & \small{1} \\
 \small{(1,19,1,1,25,1,1,55,1,1,97,1,1,199,1,1)} & \small{22} \\
 \small{(1,19,1,1,25,1,1,67,1,1,97,1,1,199,1,1)} & \small{2} \\
 \small{(1,19,1,1,25,1,1,79,1,1,97,1,1,199,1,1)} & \small{1} \\
 \small{(1,19,1,1,34,1,1,43,1,1,67,1,1,118,1,1)} & \small{3} \\
 \small{(1,19,1,25,7,1,23,19,1,29,7,25,35,19,39,41)} & \small{1} \\
 \small{(1,21,1,3,1,3,1,21,1,3,1,3,1,21,1,3)} & \small{2} \\
 \small{(1,21,1,3,7,3,1,21,1,3,7,3,1,21,1,3)} & \small{15} \\
 \small{(1,21,1,3,7,3,1,33,1,3,7,3,1,39,1,3)} & \small{1} \\
 \small{(1,21,1,3,7,3,1,57,1,3,7,3,1,21,1,3)} & \small{2} \\
 \small{(1,21,1,3,25,3,1,45,1,3,82,3,1,147,1,3)} & \small{3} \\
 \small{(1,21,1,3,34,3,1,45,1,3,82,3,1,138,1,3)} & \small{6} \\
 \small{(1,21,1,7,1,13,1,25,1,17,1,23,1,39,1,37)} & \small{1} \\
 \small{(1,21,1,7,7,3,1,25,1,3,7,7,1,21,1,7)} & \small{2} \\
 \small{(1,21,1,7,16,3,1,37,1,3,37,7,1,66,1,7)} & \small{1} \\
 \small{(1,22,1,1,4,1,1,22,1,1,4,1,1,22,1,1)} & \small{9} \\
 \small{(1,22,1,1,4,1,1,34,1,1,4,1,1,22,1,1)} & \small{1} \\
 \small{(1,22,1,1,4,1,1,34,1,1,4,1,1,58,1,1)} & \small{5} \\
 \small{(1,22,1,1,4,1,1,46,1,1,4,1,1,40,1,1)} & \small{1} \\
 \small{(1,22,1,1,4,1,1,46,1,1,4,1,1,58,1,1)} & \small{1} \\
 \small{(1,22,1,1,4,1,1,58,1,1,4,1,1,40,1,1)} & \small{1} \\
 \small{(1,24,1,3,4,3,1,24,1,3,4,3,1,24,1,3)} & \small{2} \\
 \small{(1,25,1,1,7,1,1,25,1,1,7,1,1,25,1,1)} & \small{39} \\
 \small{(1,25,1,1,7,1,1,37,1,1,7,1,1,25,1,1)} & \small{4} \\
 \small{(1,25,1,1,7,1,1,37,1,1,7,1,1,43,1,1)} & \small{3} \\
 \small{(1,25,1,1,7,1,1,37,1,1,7,1,1,61,1,1)} & \small{7} \\
 \small{(1,25,1,1,7,1,1,61,1,1,7,1,1,61,1,1)} & \small{1} \\
 \small{(1,25,1,1,7,1,1,61,1,1,7,1,1,169,1,1)} & \small{2} \\
 \small{(1,25,1,1,16,1,1,37,1,1,37,1,1,34,1,1)} & \small{2} \\
 \small{(1,25,1,1,16,1,1,37,1,1,37,1,1,70,1,1)} & \small{10} \\
 \small{(1,25,1,1,16,1,1,49,1,1,37,1,1,88,1,1)} & \small{1} \\
 \small{(1,25,1,1,16,1,1,49,1,1,37,1,1,106,1,1)} & \small{7} \\
 \small{(1,25,1,1,16,1,1,49,1,1,52,1,1,124,1,1)} & \small{4} \\
 \small{(1,25,1,1,16,1,1,97,1,1,52,1,1,142,1,1)} & \small{1} \\
 \small{(1,25,1,1,25,1,1,37,1,1,52,1,1,79,1,1)} & \small{3} \\
 \small{(1,25,1,1,25,1,1,49,1,1,67,1,1,133,1,1)} & \small{1} \\
 \small{(1,25,1,1,25,1,1,61,1,1,97,1,1,205,1,1)} & \small{4} \\
 \small{(1,25,1,1,25,1,1,73,1,1,97,1,1,205,1,1)} & \small{2} \\
 \small{(1,25,1,1,25,1,1,85,1,1,97,1,1,223,1,1)} & \small{2} \\
 \small{(1,25,1,1,25,1,1,85,1,1,127,1,1,349,1,1)} & \small{2} \\
 \small{(1,25,1,1,34,1,1,37,1,1,67,1,1,88,1,1)} & \small{2} \\
 \small{(1,25,1,1,34,1,1,61,1,1,97,1,1,178,1,1)} & \small{1} \\
 \small{(1,25,1,1,34,1,1,61,1,1,112,1,1,214,1,1)} & \small{1} \\
 \small{(1,25,1,1,43,1,1,49,1,1,82,1,1,133,1,1)} & \small{1} \\
 \small{(1,27,1,3,7,3,1,27,1,3,7,3,1,27,1,3)} & \small{6} \\
 \small{(1,27,1,3,16,3,1,63,1,3,37,3,1,108,1,3)} & \small{4} \\
 \small{(1,27,1,3,34,3,1,51,1,3,82,3,1,144,1,3)} & \small{2} \\
 \small{(1,27,1,3,43,3,1,51,1,3,97,3,1,171,1,3)} & \small{2} \\
 \small{(1,27,1,3,43,3,1,63,1,3,142,3,1,261,1,3)} & \small{2} \\
 \small{(1,28,1,1,4,1,1,52,1,1,4,1,1,82,1,1)} & \small{1} \\
 \small{(1,31,1,1,7,1,1,31,1,1,7,1,1,31,1,1)} & \small{5} \\
 \small{(1,31,1,1,7,1,1,43,1,1,7,1,1,31,1,1)} & \small{1} \\
 \small{(1,31,1,1,7,1,1,43,1,1,7,1,1,67,1,1)} & \small{6} \\
 \small{(1,31,1,1,7,1,1,55,1,1,7,1,1,103,1,1)} & \small{1} \\
 \small{(1,31,1,1,7,1,1,67,1,1,7,1,1,31,1,1)} & \small{1} \\
 \small{(1,31,1,1,16,1,1,55,1,1,37,1,1,112,1,1)} & \small{1} \\
 \small{(1,31,1,1,16,1,1,55,1,1,52,1,1,130,1,1)} & \small{2} \\
 \small{(1,31,1,1,16,1,1,115,1,1,52,1,1,166,1,1)} & \small{1} \\
 \small{(1,31,1,1,25,1,1,67,1,1,97,1,1,211,1,1)} & \small{2} \\
 \small{(1,31,1,1,25,1,1,79,1,1,82,1,1,175,1,1)} & \small{1} \\
 \small{(1,31,1,1,25,1,1,91,1,1,127,1,1,355,1,1)} & \small{2} \\
 \small{(1,33,1,3,52,3,1,69,1,3,112,3,1,204,1,3)} & \small{1} \\
 \small{(1,33,1,3,52,3,1,69,1,3,127,3,1,222,1,3)} & \small{1} \\
 \small{(1,37,1,1,61,1,1,73,1,1,127,1,1,199,1,1)} & \small{1} \\
 \small{(6,3,8,3,10,8,12,15,14,24,21,35,35,48,58,68)} & \small{225} \\
 \small{(6,3,8,3,19,8,12,15,14,24,21,35,35,57,58,68)} & \small{3} \\
 \small{(6,3,8,3,28,8,12,15,14,24,21,35,35,66,58,68)} & \small{1} \\
 \small{(6,3,8,7,10,8,12,19,14,24,21,39,35,48,58,72)} & \small{9} \\
 \small{(6,3,8,7,10,28,12,19,14,38,21,39,35,66,58,92)} & \small{1} \\
 \small{(6,3,15,3,10,8,12,15,14,31,21,35,35,48,58,68)} & \small{21} \\
 \small{(6,3,15,3,19,8,12,15,14,31,21,51,35,57,58,68)} & \small{1} \\
 \small{(6,6,8,11,13,18,23,30,40,52,69,91,120,159,210,278)} & \small{493} \\
 \small{(6,6,8,15,13,18,23,34,40,52,69,95,120,159,210,282)} & \small{2} \\
 \small{(6,6,8,19,13,18,23,38,40,52,69,99,120,159,210,286)} & \small{4} \\
 \small{(6,6,8,27,31,18,23,30,40,52,69,107,154,177,210,278)} & \small{4} \\
 \small{(6,6,8,27,49,18,23,30,40,52,69,107,154,195,210,278)} & \small{1} \\
 \small{(6,7,8,13,19,26,34,47,66,92,126,173,239,331,457,630)} & \small{111} \\
 \small{(6,7,8,17,19,26,34,51,66,92,126,177,239,331,457,634)} & \small{2} \\
 \small{(6,9,8,3,16,8,12,21,14,24,27,35,35,54,58,68)} & \small{12} \\
 \small{(6,9,8,3,25,8,12,21,14,24,27,35,35,63,58,68)} & \small{2} \\
 \small{(6,9,8,11,16,8,12,29,14,24,27,43,35,54,58,76)} & \small{2} \\
 \small{(6,9,8,11,16,18,23,33,40,52,72,91,120,162,210,278)} & \small{136} \\
 \small{(6,9,15,3,16,8,12,21,14,31,27,35,35,54,58,68)} & \small{6} \\
 \small{(6,9,15,3,25,8,12,21,14,31,27,51,35,63,58,68)} & \small{1} \\
 \small{(6,10,15,21,31,46,67,98,144,211,309,453,664,973,1426,2090)} & \small{96} \\
 \small{(6,12,8,11,13,18,23,36,40,52,69,91,120,165,210,278)} & \small{5} \\
 \small{(6,13,15,21,28,46,67,101,144,211,306,453,664,976,1426,2090)} & \small{26} \\
 \small{(6,13,29,37,70,126,177,293,495,757,1212,1989,3163,5044,8171,13130)} & \small{1} \\
 \small{(6,13,43,97,190,436,914,1981,4213,8863,18591,38689,80258,166108,342685,705076)} & \small{1} \\
 \small{(6,13,43,97,208,446,936,2005,4291,9031,18861,39233,81431,168448,347359,714366)} & \small{1} \\
 \small{(6,13,43,97,208,446,947,2017,4317,9073,18936,39393,81771,169096,348632,716866)} & \small{1} \\
 \small{(6,15,8,3,16,8,12,27,14,24,27,35,35,60,58,68)} & \small{4} \\
 \small{(6,15,8,11,16,18,23,39,40,52,72,91,120,168,210,278)} & \small{16} \\
 \small{(6,15,8,11,25,18,23,51,40,52,102,91,120,213,210,278)} & \small{10} \\
 \small{(6,15,15,3,16,8,12,27,14,31,27,35,35,60,58,68)} & \small{2} \\
 \small{(6,15,15,11,16,18,23,39,40,59,72,91,120,168,210,278)} & \small{2} \\
 \small{(6,15,15,31,37,68,89,151,209,339,486,767,1123,1743,2585,3972)} & \small{77} \\
 \small{(6,15,15,39,46,98,133,235,326,549,801,1319,1973,3192,4865,7782)} & \small{2} \\
 \small{(6,15,15,43,46,78,100,191,261,423,606,1003,1480,2328,3478,5486)} & \small{2} \\
 \small{(6,15,15,43,46,88,111,215,287,493,696,1179,1718,2796,4162,6696)} & \small{2} \\
 \small{(6,15,15,43,55,88,111,203,300,465,696,1131,1735,2697,4143,6476)} & \small{2} \\
 \small{(6,15,15,51,73,118,177,311,456,731,1161,1923,3010,4839,7734,12526)} & \small{1} \\
 \small{(6,16,29,49,67,116,188,300,469,771,1239,1985,3180,5137,8247,13244)} & \small{2} \\
 \small{(6,19,29,37,55,86,144,239,378,575,891,1397,2211,3511,5549,8710)} & \small{2} \\
 \small{(6,19,43,61,100,136,221,359,625,1065,1761,2893,4676,7624,12446,20540)} & \small{1} \\
 \small{(6,19,43,61,100,146,232,371,638,1093,1821,3005,4897,8002,13092,21590)} & \small{1} \\
 \small{(6,19,43,81,181,416,892,1915,4018,8485,17811,37105,77045,159643,329708,679216)} & \small{1} \\
 \small{(6,19,43,81,181,416,903,1915,4057,8527,17901,37297,77436,160381,331247,682236)} & \small{1} \\
 \small{(6,19,50,89,217,496,1068,2239,4681,9878,20586,42761,88384,182287,374776,768616)} & \small{1} \\
 \small{(6,19,50,97,208,466,991,2131,4486,9416,19671,40929,84780,175060,360469,740426)} & \small{1} \\
 \small{(6,19,50,97,208,466,1002,2131,4499,9444,19716,41009,84933,175366,361077,741606)} & \small{1} \\
 \small{(6,19,50,97,217,476,1024,2167,4564,9570,19986,41521,85953,177391,365105,749576)} & \small{4} \\
 \small{(6,21,8,11,16,18,23,45,40,52,72,91,120,174,210,278)} & \small{1} \\
 \small{(6,21,8,11,25,18,23,57,40,52,102,91,120,219,210,278)} & \small{4} \\
 \small{(6,21,29,47,64,108,166,289,456,731,1146,1839,2925,4710,7544,12092)} & \small{1} \\
 \small{(6,21,29,47,64,108,177,301,482,759,1206,1935,3112,5016,8076,12952)} & \small{2} \\
 \small{(6,21,29,59,73,128,199,353,560,913,1461,2411,3894,6339,10299,16816)} & \small{2} \\
 \small{(6,21,29,59,82,138,210,377,599,983,1581,2651,4268,7014,11420,18806)} & \small{1} \\
 \small{(6,21,29,63,82,148,221,397,625,1067,1701,2847,4608,7662,12503,20692)} & \small{1} \\
 \small{(6,21,29,75,109,178,287,521,833,1389,2331,3979,6631,11109,18640,31386)} & \small{1} \\
 \small{(6,21,36,55,82,128,210,373,625,1018,1656,2695,4370,7176,11819,19412)} & \small{1} \\
 \small{(6,21,43,71,118,178,309,505,898,1487,2541,4263,7192,12072,20369,34362)} & \small{1} \\
 \small{(6,21,64,115,253,548,1156,2457,5162,10762,22371,46259,95371,196041,401794,821628)} & \small{1} \\
 \small{(6,21,64,115,253,558,1189,2517,5279,11000,22836,47203,97224,199623,408824,835378)} & \small{1} \\
 \small{(6,25,50,81,142,246,419,729,1249,2164,3777,6529,11272,19564,33878,58574)} & \small{1} \\
 \small{(6,25,57,117,253,536,1145,2429,5084,10613,22056,45653,94181,193651,397139,812560)} & \small{1} \\
 \small{(6,27,8,3,16,8,12,51,14,24,27,35,35,108,58,68)} & \small{1} \\
 \small{(6,27,15,11,16,18,23,51,40,59,72,91,120,180,210,278)} & \small{1} \\
 \small{(6,27,29,3,16,8,12,39,14,59,27,35,35,72,58,68)} & \small{1} \\
 \small{(6,33,8,11,43,18,23,81,40,52,147,91,120,285,210,278)} & \small{1} \\
 \small{(11,3,8,3,10,13,12,15,14,24,26,35,35,48,58,73)} & \small{8} \\
 \small{(11,3,8,7,19,13,12,19,27,38,26,39,52,75,77,77)} & \small{8} \\
 \small{(11,3,8,11,19,13,12,23,27,38,26,43,52,75,77,81)} & \small{2} \\
 \small{(11,3,15,3,10,13,12,27,14,31,26,35,52,48,77,73)} & \small{10} \\
 \small{(11,3,15,3,19,13,12,27,14,45,26,51,52,57,96,73)} & \small{1} \\
 \small{(11,3,29,3,10,33,12,15,14,45,26,35,35,48,58,93)} & \small{1} \\
 \small{(11,3,36,3,10,13,12,15,14,52,26,35,35,48,58,73)} & \small{2} \\
 \small{(11,6,8,11,13,23,23,30,40,52,74,91,120,159,210,283)} & \small{5} \\
 \small{(11,6,8,15,22,23,23,34,53,66,74,95,137,186,229,287)} & \small{1} \\
 \small{(11,6,8,19,22,23,23,38,53,66,74,99,137,186,229,291)} & \small{1} \\
 \small{(11,6,15,11,13,23,23,42,40,59,74,91,137,159,229,283)} & \small{2} \\
 \small{(11,9,8,3,10,13,23,21,14,24,26,51,52,54,58,73)} & \small{3} \\
 \small{(11,9,8,7,19,23,23,25,27,52,56,71,69,99,134,167)} & \small{1} \\
 \small{(11,9,8,19,16,23,34,33,53,66,77,115,137,180,248,303)} & \small{41} \\
 \small{(11,9,8,19,16,33,34,33,53,66,77,115,137,180,248,313)} & \small{1} \\
 \small{(11,9,15,19,25,43,45,81,105,143,227,291,443,621,875,1283)} & \small{9} \\
 \small{(11,9,15,27,25,53,67,93,157,199,317,459,647,999,1407,2093)} & \small{6} \\
 \small{(11,9,22,19,25,53,45,81,118,150,227,307,460,621,894,1313)} & \small{1} \\
 \small{(11,13,22,49,73,101,155,273,443,680,1061,1729,2789,4405,6955,11129)} & \small{4} \\
 \small{(11,15,8,19,25,23,45,51,53,94,107,131,205,231,305,443)} & \small{5} \\
 \small{(11,15,15,19,25,43,45,87,105,143,227,291,443,627,875,1283)} & \small{2} \\
 \small{(11,15,15,27,25,53,67,99,157,199,317,459,647,1005,1407,2093)} & \small{1} \\
 \small{(11,16,22,37,58,91,144,224,352,554,869,1365,2143,3364,5283,8295)} & \small{23} \\
 \small{(11,16,29,45,76,121,188,308,495,799,1289,2077,3350,5398,8703,14025)} & \small{2} \\
 \small{(11,16,29,49,76,121,199,324,521,841,1364,2209,3571,5776,9349,15129)} & \small{3} \\
 \small{(11,16,29,49,76,131,199,324,534,855,1379,2225,3622,5866,9482,15359)} & \small{1} \\
 \small{(11,16,29,49,76,131,199,324,534,855,1394,2225,3622,5866,9501,15359)} & \small{1} \\
 \small{(11,18,29,47,76,123,199,322,521,843,1364,2207,3571,5778,9349,15127)} & \small{72} \\
 \small{(11,18,36,59,94,153,265,434,716,1186,1979,3275,5424,9000,14935,24761)} & \small{6} \\
 \small{(11,18,36,59,94,163,265,434,729,1200,1994,3291,5492,9090,15087,25031)} & \small{1} \\
 \small{(11,18,36,59,94,173,265,446,742,1242,2069,3403,5696,9450,15752,26161)} & \small{1} \\
 \small{(11,18,36,59,94,183,265,458,755,1284,2114,3515,5900,9810,16341,27271)} & \small{1} \\
 \small{(11,18,36,67,94,163,276,458,768,1256,2114,3507,5849,9720,16208,26991)} & \small{1} \\
 \small{(11,19,29,57,82,121,199,351,573,911,1466,2409,3928,6364,10337,16869)} & \small{1} \\
 \small{(11,19,36,65,91,131,221,399,664,1058,1721,2865,4710,7687,12636,20899)} & \small{2} \\
 \small{(11,21,15,19,25,43,45,93,105,143,227,291,443,633,875,1283)} & \small{2} \\
 \small{(11,21,29,19,25,53,45,93,131,171,227,307,460,633,913,1353)} & \small{1} \\
 \small{(11,21,29,75,118,193,298,533,911,1529,2516,4283,7277,12306,20635,34801)} & \small{1} \\
 \small{(11,21,36,83,127,193,320,605,1028,1704,2861,4931,8433,14277,24188,41241)} & \small{1} \\
 \small{(11,21,36,91,145,233,430,749,1275,2166,3746,6507,11306,19533,33859,58681)} & \small{1} \\
 \small{(11,22,29,45,76,121,199,326,521,841,1364,2205,3571,5782,9349,15125)} & \small{2} \\
 \small{(11,22,36,53,94,151,243,410,677,1114,1844,3045,5016,8284,13681,22575)} & \small{2} \\
 \small{(11,22,36,57,94,171,265,450,742,1226,2039,3401,5645,9382,15600,25939)} & \small{1} \\
 \small{(11,24,36,63,103,173,287,484,807,1354,2264,3791,6342,10617,17766,29737)} & \small{7} \\
 \small{(11,24,43,75,121,213,364,620,1054,1809,3089,5275,9011,15405,26316,44961)} & \small{3} \\
 \small{(11,27,36,3,16,23,12,51,14,80,32,35,52,72,96,83)} & \small{1} \\
 \small{(11,30,71,151,337,703,1464,3094,6371,13149,27044,55335,112915,229827,466698,946067)} & \small{1} \\
 \small{(11,31,78,149,334,711,1530,3155,6553,13504,27671,56581,115261,234148,474963,961455)} & \small{1} \\
 \small{(11,31,78,165,343,731,1552,3227,6670,13714,28076,57333,116689,236875,479922,970695)} & \small{1} \\
 \small{(11,31,78,165,343,741,1563,3239,6696,13756,28166,57493,116961,237379,480872,972445)} & \small{1} \\
 \small{(11,33,99,187,388,873,1761,3605,7411,15053,30611,61963,124951,252060,507643,1020581)} & \small{1} \\
 \small{(16,3,8,3,10,18,12,15,14,24,31,35,35,48,58,78)} & \small{1} \\
 \small{(16,3,8,7,19,18,12,19,27,38,31,39,52,75,77,82)} & \small{4} \\
 \small{(16,3,8,7,37,18,12,19,27,52,31,39,52,93,96,82)} & \small{1} \\
 \small{(16,3,8,11,28,18,12,23,40,52,31,43,69,102,96,86)} & \small{2} \\
 \small{(16,3,15,3,10,28,12,15,14,31,46,35,35,48,58,108)} & \small{2} \\
 \small{(16,6,8,11,13,28,23,30,40,52,79,91,120,159,210,288)} & \small{2} \\
 \small{(16,6,8,19,31,28,23,38,66,80,79,99,154,213,248,296)} & \small{1} \\
 \small{(16,9,8,3,16,18,12,21,14,24,37,35,35,54,58,78)} & \small{5} \\
 \small{(16,9,8,11,16,28,23,33,40,52,82,91,120,162,210,288)} & \small{4} \\
 \small{(16,9,15,19,25,48,45,81,105,143,232,291,443,621,875,1288)} & \small{2} \\
 \small{(16,9,15,27,16,48,45,57,105,87,172,203,239,414,419,668)} & \small{2} \\
 \small{(16,9,15,27,25,58,67,93,157,199,322,459,647,999,1407,2098)} & \small{2} \\
 \small{(16,9,22,3,16,18,12,21,14,38,37,35,35,54,58,78)} & \small{1} \\
 \small{(16,21,29,19,25,78,45,93,131,171,232,323,477,633,913,1378)} & \small{1} \\
 \small{(16,21,43,63,106,178,287,481,794,1333,2212,3695,6172,10290,17196,28702)} & \small{1} \\
 \small{(16,21,43,71,124,208,353,613,1028,1767,3007,5127,8756,14916,25461,43412)} & \small{1} \\
 \small{(16,21,43,87,133,218,375,649,1093,1879,3232,5527,9504,16257,27855,47802)} & \small{1} \\
 \small{(16,22,29,53,85,146,243,386,638,1051,1729,2853,4676,7681,12636,20770)} & \small{1} \\
 \small{(16,22,36,61,103,176,287,482,807,1352,2269,3789,6342,10615,17766,29740)} & \small{1} \\
 \small{(16,22,36,65,112,186,309,522,885,1492,2509,4225,7124,12010,20236,34094)} & \small{1} \\
 \small{(16,22,43,73,121,206,353,606,1028,1751,2989,5097,8688,14809,25252,43054)} & \small{1} \\
 \small{(16,24,36,63,103,178,298,496,833,1396,2344,3935,6597,11067,18564,31142)} & \small{4} \\
 \small{(16,24,43,75,130,218,375,644,1106,1893,3244,5563,9538,16350,28026,48046)} & \small{4} \\
 \small{(16,24,43,83,139,238,408,716,1236,2131,3679,6371,11017,19041,32909,56906)} & \small{1} \\
 \small{(16,24,50,91,148,258,463,824,1418,2474,4324,7579,13210,23082,40357,70526)} & \small{1} \\
 \small{(16,25,50,89,151,266,463,813,1418,2472,4327,7561,13210,23083,40338,70494)} & \small{1} \\
 \small{(16,27,50,79,142,238,408,715,1223,2124,3667,6335,10949,18918,32700,56502)} & \small{2} \\
 \small{(16,27,50,87,142,248,430,751,1288,2236,3907,6743,11731,20340,35284,61252)} & \small{1} \\
 \small{(16,28,36,61,103,176,309,512,846,1422,2389,4029,6784,11377,19096,32080)} & \small{1} \\
 \small{(16,28,36,65,112,196,331,552,924,1576,2674,4529,7651,12934,21889,37064)} & \small{1} \\
 \small{(16,28,43,69,121,206,353,608,1028,1751,2989,5093,8688,14815,25252,43050)} & \small{1} \\
 \small{(16,28,43,73,121,216,375,636,1080,1849,3169,5433,9300,15913,27228,46604)} & \small{4} \\
 \small{(16,28,43,73,130,226,386,660,1132,1947,3349,5753,9878,16966,29147,50074)} & \small{2} \\
 \small{(16,30,43,79,130,228,397,682,1171,2019,3469,5983,10303,17742,30553,52612)} & \small{4} \\
 \small{(16,30,43,91,148,258,452,794,1379,2411,4204,7355,12819,22350,38989,68046)} & \small{1} \\
 \small{(16,30,50,91,157,278,485,854,1496,2628,4609,8091,14196,24915,43720,76726)} & \small{3} \\
 \small{(16,30,50,99,166,298,529,938,1652,2922,5164,9139,16151,28560,50503,89306)} & \small{1} \\
 \small{(16,30,50,99,175,308,540,962,1704,3020,5344,9475,16797,29775,52764,93516)} & \small{6} \\
 \small{(16,31,71,181,370,786,1629,3395,6943,14281,29131,59237,120208,243382,491892,992570)} & \small{1} \\
 \small{(16,31,78,173,370,776,1607,3335,6878,14092,28786,58621,119035,241222,487940,985400)} & \small{1} \\
 \small{(16,31,78,189,397,836,1739,3575,7320,14946,30376,61533,124373,250915,505458,1017020)} & \small{2} \\
 \small{(16,31,78,189,406,846,1761,3623,7411,15114,30691,62109,125410,252778,508783,1022910)} & \small{1} \\
 \small{(16,31,85,193,406,836,1728,3555,7281,14869,30211,61201,123761,249790,503463,1013444)} & \small{1} \\
 \small{(16,31,92,201,415,866,1794,3651,7489,15212,30826,62329,125750,253327,509695,1024414)} & \small{2} \\
 \small{(16,31,92,201,424,876,1816,3699,7567,15366,31096,62809,126600,254812,512298,1028924)} & \small{2} \\
 \small{(16,33,50,95,151,248,452,769,1327,2334,4042,7007,12207,21201,36861,64152)} & \small{1} \\
 \small{(16,33,50,95,151,248,452,769,1340,2348,4072,7071,12326,21399,37260,64852)} & \small{1} \\
 \small{(16,34,50,77,139,236,419,734,1262,2178,3769,6525,11306,19591,33916,58720)} & \small{1} \\
 \small{(16,34,50,81,148,266,441,786,1366,2374,4144,7233,12564,21922,38210,66554)} & \small{1} \\
 \small{(16,34,50,81,148,266,463,810,1405,2444,4279,7489,13074,22822,39844,69574)} & \small{3} \\
 \small{(16,34,50,81,157,276,474,834,1457,2542,4474,7841,13703,24001,42048,73624)} & \small{1} \\
 \small{(16,34,50,89,148,266,463,810,1405,2458,4294,7497,13074,22876,39920,69734)} & \small{1} \\
 \small{(16,34,50,89,148,266,474,834,1444,2528,4414,7737,13550,23722,41497,72614)} & \small{1} \\
 \small{(16,34,50,89,157,276,485,858,1496,2626,4609,8089,14196,24919,43720,76724)} & \small{1} \\
 \small{(16,34,57,97,166,306,529,930,1652,2927,5149,9137,16151,28546,50465,89254)} & \small{1} \\
 \small{(16,36,57,107,184,338,606,1088,1938,3475,6214,11131,19925,35676,63860,114326)} & \small{2} \\
 \small{(16,36,57,107,193,348,617,1112,1990,3573,6409,11499,20622,36999,66368,119056)} & \small{8} \\
 \small{(16,37,50,93,178,286,518,941,1613,2878,5122,8957,15879,28108,49496,87550)} & \small{1} \\
 \small{(16,37,57,97,187,326,573,1041,1847,3277,5872,10465,18633,33265,59319,105734)} & \small{5} \\
 \small{(16,37,78,165,361,756,1596,3305,6813,13994,28606,58277,118423,240103,485888,981640)} & \small{1} \\
 \small{(16,37,78,173,370,776,1618,3353,6904,14148,28891,58813,119409,241912,489194,987700)} & \small{4} \\
 \small{(16,37,85,181,388,816,1695,3497,7177,14659,29836,60549,122588,247708,499720,1006760)} & \small{2} \\
 \small{(16,37,85,193,406,836,1728,3561,7281,14869,30211,61201,123761,249796,503463,1013444)} & \small{1} \\
 \small{(16,37,92,201,406,866,1783,3645,7476,15170,30781,62217,125580,253018,509144,1023474)} & \small{1} \\
 \small{(16,37,92,201,415,866,1794,3657,7489,15212,30826,62329,125750,253333,509695,1024414)} & \small{1} \\
 \small{(16,37,92,201,424,876,1816,3705,7567,15366,31096,62809,126600,254818,512298,1028924)} & \small{1} \\
 \small{(16,39,92,191,406,848,1761,3619,7385,15046,30541,61807,124849,251742,506902,1019532)} & \small{1} \\
 \small{(16,39,92,191,406,858,1772,3631,7424,15116,30661,62031,125240,252444,508156,1021722)} & \small{1} \\
 \small{(16,39,92,207,451,898,1860,3811,7749,15676,31621,63743,128215,257529,516991,1036982)} & \small{1} \\
 \small{(16,39,99,199,415,888,1827,3727,7606,15431,31231,63063,127025,255567,513647,1031252)} & \small{1} \\
 \small{(16,43,78,185,388,806,1695,3471,7151,14582,29701,60313,122112,246904,498200,1004074)} & \small{1} \\
 \small{(16,43,78,185,388,806,1695,3483,7164,14610,29761,60409,122299,247228,498808,1005154)} & \small{1} \\
 \small{(16,43,85,177,388,806,1684,3483,7125,14589,29686,60257,122078,246742,498010,1003674)} & \small{2} \\
 \small{(16,43,85,185,397,816,1706,3507,7190,14687,29866,60617,122690,247885,500043,1007304)} & \small{6} \\
 \small{(16,43,92,193,415,856,1772,3639,7424,15114,30646,62001,125189,252331,507947,1021344)} & \small{7} \\
 \small{(16,45,92,195,433,868,1805,3701,7515,15312,30961,62579,126226,254079,511101,1026776)} & \small{1} \\
 \small{(16,45,99,195,424,888,1816,3725,7580,15375,31141,62851,126668,254970,512507,1029296)} & \small{1} \\
 \small{(16,45,99,195,433,878,1805,3725,7541,15347,31066,62691,126481,254529,511785,1028126)} & \small{1} \\
 \small{(16,45,99,195,433,888,1816,3737,7580,15403,31171,62883,126787,255105,512754,1029776)} & \small{1} \\
 \small{(16,45,106,203,442,918,1860,3821,7736,15648,31636,63675,128113,257454,516706,1036526)} & \small{3} \\
 \small{(16,45,106,211,442,928,1893,3857,7827,15802,31846,64083,128776,258534,518625,1039616)} & \small{1} \\
 \small{(16,45,106,211,451,928,1893,3869,7827,15816,31876,64115,128861,258669,518815,1039976)} & \small{1} \\
 \small{(21,3,8,11,28,43,12,23,40,52,51,43,69,102,96,131)} & \small{1} \\
 \small{(21,3,8,23,28,43,12,27,40,52,36,55,69,102,96,115)} & \small{1} \\
 \small{(21,3,8,23,73,23,12,39,79,94,36,71,154,237,153,135)} & \small{1} \\
 \small{(21,9,22,19,25,73,45,81,118,150,237,307,460,621,894,1333)} & \small{1} \\
 \small{(21,25,50,105,244,501,1057,2173,4577,9318,19146,39033,79833,162610,331190,672481)} & \small{1} \\
 \small{(21,25,57,121,280,591,1277,2677,5617,11663,24171,49785,102239,209338,427786,871971)} & \small{1} \\
 \small{(21,31,78,165,343,711,1453,2987,6085,12370,24996,50757,102443,206959,416576,839475)} & \small{1} \\
 \small{(21,31,78,165,343,711,1475,2999,6111,12370,25041,50581,101967,205249,413080,829895)} & \small{1} \\
 \small{(21,36,64,115,211,383,694,1256,2276,4126,7479,13555,24566,44523,80694,146251)} & \small{7} \\
 \small{(21,37,50,101,187,321,573,1025,1808,3242,5787,10261,18310,32635,58084,103525)} & \small{1} \\
 \small{(21,37,50,109,196,321,606,1061,1873,3410,6027,10765,19347,34426,61561,110125)} & \small{1} \\
 \small{(21,37,50,109,196,321,606,1073,1899,3452,6132,10941,19653,35074,62720,112305)} & \small{1} \\
 \small{(21,37,57,105,196,341,617,1113,1977,3557,6387,11433,20520,36802,65969,118329)} & \small{1} \\
 \small{(21,37,57,113,205,361,661,1185,2133,3865,6957,12545,22645,40825,73645,132849)} & \small{8} \\
 \small{(21,37,78,189,388,821,1706,3521,7190,14708,29916,60669,122826,248140,500461,1008065)} & \small{1} \\
 \small{(21,37,85,197,424,871,1816,3713,7580,15373,31161,62901,126804,255178,512963,1030075)} & \small{1} \\
 \small{(21,37,92,197,415,861,1783,3653,7463,15170,30756,62197,125529,252937,509011,1023205)} & \small{1} \\
 \small{(21,37,92,197,415,871,1794,3677,7502,15254,30891,62453,125971,253729,510379,1025615)} & \small{1} \\
 \small{(21,37,92,197,424,871,1816,3713,7580,15366,31131,62853,126702,254980,512621,1029455)} & \small{1} \\
 \small{(21,37,99,205,424,881,1827,3725,7593,15401,31161,62909,126787,255124,512811,1029805)} & \small{1} \\
 \small{(21,37,99,205,442,911,1882,3833,7801,15765,31821,64061,128793,258580,518758,1039935)} & \small{1} \\
 \small{(21,39,64,127,223,403,749,1351,2471,4518,8217,15007,27371,49899,91049,166047)} & \small{3} \\
 \small{(21,39,71,131,241,443,815,1499,2757,5071,9327,17155,31553,58035,106743,196331)} & \small{30} \\
 \small{(21,40,92,193,409,851,1761,3612,7385,15030,30504,61745,124730,251527,506522,1018839)} & \small{1} \\
 \small{(21,42,99,211,436,903,1849,3770,7671,15543,31404,63331,127501,256362,514939,1033451)} & \small{1} \\
 \small{(21,43,71,201,415,841,1783,3639,7424,15093,30651,61977,125121,252313,507852,1021169)} & \small{1} \\
 \small{(21,43,92,201,415,861,1783,3651,7450,15156,30726,62137,125427,252763,508688,1022649)} & \small{3} \\
 \small{(21,43,99,209,433,901,1849,3771,7671,15541,31401,63329,127501,256363,514939,1033449)} & \small{13} \\
 \small{(21,43,99,209,442,911,1882,3819,7762,15695,31671,63793,128300,257722,517257,1037359)} & \small{2} \\
 \small{(21,43,106,217,442,911,1871,3807,7736,15646,31581,63641,128045,257290,516516,1036119)} & \small{1} \\
 \small{(21,43,106,217,460,941,1926,3903,7905,15940,32091,64489,129473,259666,520487,1042709)} & \small{2} \\
 \small{(21,45,99,195,424,833,1805,3729,7489,15221,31071,62771,126413,254376,511747,1027973)} & \small{1} \\
 \small{(21,45,99,207,424,893,1849,3757,7645,15515,31356,63247,127365,256122,514540,1032797)} & \small{1} \\
 \small{(21,45,99,207,433,903,1849,3769,7671,15543,31401,63327,127501,256365,514939,1033447)} & \small{1} \\
 \small{(21,45,99,211,442,903,1860,3797,7697,15599,31506,63491,127790,256860,515775,1034871)} & \small{1} \\
 \small{(21,45,106,219,451,933,1904,3869,7840,15830,31896,64171,128929,258777,519005,1040281)} & \small{2} \\
 \small{(21,45,106,223,469,953,1937,3937,7957,16012,32196,64703,129796,260145,521285,1044037)} & \small{1} \\
 \small{(21,45,113,231,478,963,1970,3985,8022,16117,32376,64983,130238,260856,522406,1045787)} & \small{1} \\
 \small{(21,45,113,231,478,963,1970,3985,8035,16131,32391,65015,130306,260964,522558,1046027)} & \small{1} \\
 \small{(21,46,92,197,418,861,1783,3662,7463,15184,30759,62245,125580,253054,509182,1023545)} & \small{1} \\
 \small{(21,46,99,209,436,901,1849,3774,7671,15541,31404,63329,127501,256366,514939,1033449)} & \small{3} \\
 \small{(21,48,106,219,454,933,1904,3872,7840,15830,31899,64171,128929,258780,519005,1040281)} & \small{2} \\
 \small{(21,51,106,211,442,933,1893,3851,7827,15788,31851,64083,128742,258522,518568,1039521)} & \small{1} \\
 \small{(21,51,106,211,460,933,1893,3875,7827,15816,31881,64115,128878,258666,518796,1040001)} & \small{1} \\
 \small{(21,51,113,227,460,963,1959,3959,7996,16075,32316,64867,130051,260610,521988,1045131)} & \small{1} \\
 \small{(21,51,113,227,469,963,1959,3971,7996,16089,32331,64883,130102,260673,522083,1045291)} & \small{1} \\
 \small{(21,51,120,235,478,973,1981,4007,8048,16166,32451,65083,130425,261168,522843,1046481)} & \small{2} \\
 \small{(26,3,8,71,46,28,12,63,105,94,41,183,171,228,153,180)} & \small{1} \\
 \small{(26,31,99,165,397,716,1486,3047,6306,12657,25781,52101,105044,211225,424309,852040)} & \small{1} \\
 \small{(26,33,85,223,460,918,1904,3865,7801,15795,31901,64127,128827,258684,518815,1039882)} & \small{1} \\
 \small{(26,33,92,231,478,968,1970,3973,8009,16138,32441,65063,130374,261114,522786,1046392)} & \small{1} \\
 \small{(26,33,106,219,433,878,1783,3569,7216,14472,29036,58267,116927,234141,468598,936986)} & \small{1} \\
 \small{(26,33,113,235,478,978,1981,3965,8048,16173,32456,65083,130408,261114,522862,1046486)} & \small{1} \\
 \small{(26,39,71,139,259,478,881,1631,3030,5617,10397,19259,35684,66117,122494,226926)} & \small{2} \\
 \small{(26,39,71,139,259,488,925,1691,3134,5855,10862,20203,37537,69717,129581,240776)} & \small{1} \\
 \small{(26,39,71,139,268,488,903,1679,3134,5827,10817,20075,37282,69276,128669,238976)} & \small{1} \\
 \small{(26,39,99,215,451,918,1904,3847,7788,15781,31826,63975,128691,258393,518245,1039082)} & \small{1} \\
 \small{(26,39,106,227,469,948,1937,3911,7931,15984,32171,64611,129643,259923,520962,1043456)} & \small{1} \\
 \small{(26,39,106,227,469,948,1937,3923,7931,15984,32171,64611,129643,259959,520962,1043456)} & \small{1} \\
 \small{(26,39,113,235,478,978,1981,3971,8048,16173,32456,65083,130408,261120,522862,1046486)} & \small{1} \\
 \small{(26,39,113,235,478,978,1981,3983,8048,16173,32456,65083,130408,261156,522862,1046486)} & \small{3} \\
 \small{(26,39,120,243,478,968,1970,3959,7996,16054,32201,64563,129541,259572,519784,1040476)} & \small{1} \\
 \small{(26,43,64,137,250,466,859,1551,2887,5370,9902,18265,33712,62260,115008,212394)} & \small{1} \\
 \small{(26,43,71,137,259,476,881,1623,3004,5587,10337,19129,35429,65599,121506,225024)} & \small{1} \\
 \small{(26,45,71,187,352,768,1640,3321,6969,14381,29201,59451,120718,243882,492595,994628)} & \small{1} \\
 \small{(26,45,78,147,277,518,969,1805,3368,6296,11762,21971,41039,76653,143185,267466)} & \small{2} \\
 \small{(26,45,85,227,469,928,1937,3917,7866,15921,32111,64483,129422,259695,520411,1042496)} & \small{1} \\
 \small{(26,45,92,207,433,858,1761,3589,7216,14486,29126,58495,117386,235341,471505,943882)} & \small{1} \\
 \small{(26,45,92,223,451,918,1904,3865,7801,15802,31871,64063,128793,258615,518625,1039662)} & \small{1} \\
 \small{(26,45,99,215,442,898,1860,3793,7684,15571,31481,63431,127671,256698,515490,1034342)} & \small{1} \\
 \small{(26,45,99,231,478,968,1970,3985,8022,16145,32441,65063,130374,261126,522805,1046412)} & \small{1} \\
 \small{(26,45,106,223,451,938,1926,3877,7866,15900,31976,64303,129201,259173,519651,1041422)} & \small{1} \\
 \small{(26,45,106,227,451,938,1915,3881,7866,15872,31946,64275,129099,259029,519461,1041026)} & \small{1} \\
 \small{(26,45,106,227,469,948,1937,3929,7931,15984,32171,64611,129643,259965,520962,1043456)} & \small{1} \\
 \small{(26,45,113,235,478,978,1981,3989,8048,16173,32456,65083,130408,261162,522862,1046486)} & \small{9} \\
 \small{(26,45,120,243,487,998,2014,4037,8126,16292,32636,65363,130833,261801,523831,1047946)} & \small{2} \\
 \small{(26,49,113,225,475,956,1959,3945,7983,16045,32282,64769,129949,260401,521703,1044604)} & \small{2} \\
 \small{(26,51,85,223,469,928,1904,3895,7853,15809,31931,64239,128946,258837,519195,1040492)} & \small{1} \\
 \small{(26,51,92,223,460,928,1904,3883,7840,15816,31916,64175,128912,258792,519024,1040252)} & \small{1} \\
 \small{(26,51,92,231,478,968,1970,3991,8035,16138,32441,65063,130374,261132,522824,1046432)} & \small{1} \\
 \small{(26,51,99,227,469,938,1937,3935,7918,15963,32156,64579,129558,259881,520829,1043166)} & \small{1} \\
 \small{(26,51,99,231,478,968,1970,3991,8035,16145,32441,65063,130374,261132,522824,1046432)} & \small{1} \\
 \small{(26,51,106,219,460,928,1893,3875,7827,15802,31886,64139,128861,258684,518853,1040016)} & \small{1} \\
 \small{(26,51,106,227,469,948,1937,3935,7931,15984,32171,64611,129643,259971,520962,1043456)} & \small{3} \\
 \small{(26,51,106,231,469,958,1948,3955,7983,16054,32276,64807,129966,260457,521760,1044762)} & \small{1} \\
 \small{(26,51,113,227,469,968,1948,3947,7983,16047,32261,64787,129915,260403,521665,1044576)} & \small{2} \\
 \small{(26,51,113,231,478,968,1970,4003,8048,16145,32456,65079,130391,261168,522843,1046452)} & \small{1} \\
 \small{(26,51,113,235,478,978,1981,3995,8048,16173,32456,65083,130408,261168,522862,1046486)} & \small{7} \\
 \small{(26,51,113,235,484,978,1981,3995,8048,16173,32462,65083,130408,261168,522862,1046486)} & \small{4} \\
 \small{(26,51,113,239,487,978,1992,4027,8087,16215,32546,65231,130612,261483,523375,1047242)} & \small{1} \\
 \small{(26,51,113,239,487,988,2003,4039,8113,16271,32621,65343,130799,261771,523793,1047892)} & \small{2} \\
 \small{(26,51,120,243,478,978,1992,4007,8061,16194,32486,65139,130510,261312,523071,1046826)} & \small{1} \\
 \small{(26,51,120,243,487,988,2003,4031,8100,16250,32576,65267,130697,261609,523527,1047496)} & \small{3} \\
 \small{(26,51,120,243,487,998,2014,4043,8126,16292,32636,65363,130833,261807,523831,1047946)} & \small{4} \\
 \small{(26,55,120,245,496,1006,2025,4063,8152,16332,32696,65445,130952,261982,524078,1048298)} & \small{2} \\
 \small{(26,58,120,245,499,1006,2025,4066,8152,16332,32699,65445,130952,261985,524078,1048298)} & \small{1} \\
 \small{(31,3,8,7,19,63,12,19,27,38,46,39,52,75,77,147)} & \small{1} \\
 \small{(31,15,8,11,10,83,34,39,40,52,151,107,137,168,210,423)} & \small{1} \\
 \small{(31,15,8,11,10,93,34,39,40,52,211,107,137,168,210,653)} & \small{1} \\
 \small{(31,21,8,11,10,73,199,69,40,52,136,203,222,192,210,433)} & \small{1} \\
 \small{(31,33,8,11,10,123,166,81,40,52,226,475,460,258,210,883)} & \small{1} \\
 \small{(31,45,127,255,487,1023,2047,4041,8191,16383,32671,65535,131071,261945,524287,1048575)} & \small{1} \\
 \small{(31,51,127,255,505,1023,2047,4083,8191,16383,32761,65535,131071,262131,524287,1048575)} & \small{1} \\
 \small{(31,54,127,255,508,1023,2047,4086,8191,16383,32764,65535,131071,262134,524287,1048575)} & \small{1} \\
 \small{(31,55,127,253,505,1011,2047,4087,8191,16381,32761,65533,131071,262135,524287,1048563)} & \small{1} \\
 \small{(31,58,127,253,508,1021,2047,4090,8191,16381,32764,65533,131071,262138,524287,1048573)} & \small{1} \\
 \small{(31,61,127,253,511,1021,2047,4093,8191,16381,32767,65533,131071,262141,524287,1048573)} & \small{23} \\
 \small{(31,63,127,255,511,1023,2047,4095,8191,16383,32767,65535,131071,262143,524287,1048575)} & \small{9}
\end{longtable}

\end{document}